\documentclass[aap,reqno]{imsart}
\setattribute{journal}{name}{} 
\arxiv{math.PR/0000000}

\usepackage{graphicx}
\usepackage{psfrag}
\usepackage{mathrsfs}

\usepackage{amsmath,amsfonts,amssymb,latexsym,epsfig, amsthm}
\usepackage{verbatim}
\usepackage{fullpage}

\RequirePackage[numbers]{natbib}
\RequirePackage[colorlinks,citecolor=blue,urlcolor=blue]{hyperref}

\newtheorem{prop}{Proposition}

\newtheorem{theorem}{Theorem}
\newtheorem{assumptions}{Assumption}
\newtheorem{remark}{Remark}

\newcommand{\wtilde}{\widetilde}
\newcommand{\drift}{\mathscr{D}}
\newcommand{\pr}{\mathrm{proj}_{\mathcal{M}}}
\newcommand{\PP}{\mathbb{P}}
\renewcommand{\P}[1]{\PP\left( #1 \right)}

\newcommand{\EE}{\mathbb{E}}
\newcommand{\E}[1]{\EE\left[ #1  \right]}

\newcommand{\RR}{\mathbb{R}}

\newcommand{\A}{\mathcal{A}}
\newcommand{\m}{\mathcal{M}}
\newcommand{\OO}{\mathcal{O}}
\renewcommand{\epsilon}{\varepsilon}
\renewcommand\phi{\varphi}
\newcommand{\Normal}{\ensuremath{\operatorname{N}}}
\newcommand{\Exponential}{\ensuremath{\operatorname{Exp}}}

\newcommand{\trace}{\ensuremath{\operatorname{Tr}}}

\newcommand{\indic}{\textbf{1}}

\newcommand{\floor}[1]{\lfloor #1 \rfloor}
\newcommand{\Gen}{\mathcal{L}}
\newcommand{\GG}{\mathcal{G}}
\newcommand{\esjd}{\textrm{ESJD}}
\newcommand{\argmax}{\textrm{argmax}}
\newcommand{\dirichlet}{\mathcal{D}}

\newcommand{\skorokhod}{\mathbf{D}}

\newcommand{\bra}[1]{\langle #1 \rangle}
\newcommand{\BK}[1]{ {\left( #1 \right)} }
\newcommand{\sqBK}[1]{ {\left[ #1 \right]} }
\newcommand{\curBK}[1]{ {\left\{ #1 \right\}} }

\newcommand{\norm}[1]{\left\Vert #1 \right\Vert}
\newcommand{\abs}[1]{\left| #1 \right|}

\begin{document}

\begin{frontmatter}

\title{Asymptotic Analysis of the Random-Walk Metropolis Algorithm on Ridged Densities}
\runtitle{MCMC}

\begin{aug}
  \author{ \fnms{Alexandros} \snm{Beskos}\thanksref{m1} \ead[label=e1]{a.beskos@ucl.ac.uk}},
  \author{\fnms{Gareth} \snm{Roberts}\thanksref{m2}
  \ead[label=e2]{gareth.o.roberts@warwick.ac.uk}},
  \author{\fnms{Alexandre} \snm{Thiery}\thanksref{m3}%
  \ead[label=e3]{a.h.thiery@nus.edu.sg}}
   \and
   \author{ \fnms{Natesh} \snm{Pillai}\thanksref{m4} \ead[label=e4]{pillai@fas.harvard.edu}},
  \runauthor{Beskos et al.}
  \affiliation{University College London\thanksmark{m1}, University of Warwick\thanksmark{m2}, \\
National University of Singapore\thanksmark{m3}, Harvard University\thanksmark{m4}}

  \address{Gower Street, London,\\ 
          \printead{e1}}

  \address{University of Warwick, Coventry,\\
          \printead{e2}}

 \address{National University of Singapore, Singapore,\\
          \printead{e3}}
          
 \address{Harvard University, Boston,\\
          \printead{e4}}

\end{aug}

\begin{abstract}
In this paper we study the asymptotic behavior of the Random-Walk Metropolis algorithm on probability densities with two different `scales', where most of the probability mass is distributed along certain key directions with the `orthogonal' directions containing relatively less mass.
Such class of probability measures arise in various applied contexts including Bayesian  inverse problems where the posterior measure concentrates on a sub-manifold when the noise variance goes to zero. 
When the target measure concentrates on a linear sub-manifold, we derive analytically a diffusion limit for the Random-Walk Metropolis Markov chain as the scale parameter goes to zero. In contrast to the existing works on scaling limits, our limiting Stochastic Differential Equation does \emph{not} in general have a constant diffusion
coefficient. 
Our results show that in some cases, the usual practice of adapting the step-size to control the acceptance probability might be sub-optimal as the optimal acceptance probability is zero (in the limit).
\end{abstract}



\end{frontmatter}


\section{Introduction}
Optimal scaling of Metropolis-Hastings (MH) algorithms in high dimensions and analysis of their asymptotic behavior has been a fruitful area of research in the last three decades.
Initiated by \cite{robe:97}, a long list of papers has been devoted to deriving the optimal scale of various local-move Markov Chain Monte-Carlo (MCMC) algorithms \cite{roberts2001optimal,robertsminimising,pasarica2010adaptively,breyer2004optimal,BedardRosenthal:2008}. A main result in most of these works is  identifying the proper scale for the proposal distribution at which the average acceptance probability is $\mathcal{O}(1)$ as a function of the dimension, and obtaining an associated diffusion limit. The limiting Stochastic Differential Equation (SDE) in all of the earlier works has a constant diffusion coefficient which uniquely characterizes its `speed measure', with the latter being controlled by the step-size of the local-proposal and the average acceptance probability. 
For the Random-Walk Metropolis algorithm (RWM), maximizing the speed of the limiting diffusion leads to an average acceptance probability of 0.234. 
Thus the theoretical asymptotic analysis provides a simple optimality criterion that can be used as a guideline for practical implementation of these algorithms.
Most of the papers cited above assume that the target measure is a product of one-dimensional densities - as this facilitates explicit calculation of the diffusion limit whose dynamics are independent of the state of other components. However,
the product form assumption is not essential, and there are now a number of recent generalisations to target measures that are not of a product form \cite{bedard2007weak,beskos2009optimal,mattingly2012diffusion,PST12}.

We adopt a different point of view in this paper, and study the behavior of the RWM algorithm on a class of target distributions that have `ridges' in certain directions. Such target distributions arise in a number of examples in applied contexts. For instance, in Bayesian inverse problems, when the variance parameter in the likelihood model is small, the posterior distribution (as parametrized by the noise variance) will concentrate on a submanifold or along a hyperplane as the noise variance goes to zero. 
In contrast to the works cited above on scaling, the parameter space in our problems need not be high dimensional. The key issue pertaining to the mixing of MCMC algorithms in these contexts is that there is a natural separation of scale in the target distribution: a larger scale along certain key directions of interest where most of the probability mass is distributed and a smaller scale in the `orthogonal directions' where there is relatively little mass. Thus in the context of previous works cited above, the role of the dimension is played by this scale. 
The contributions of this work are summarised as follows:
\begin{itemize}
\item[i.]
Motivated by inverse problems, we first describe a natural class of target distributions which have two scales (of magnitudes $\mathcal{O}(1)$ and $\epsilon \ll 1$). Next, we study the RWM with step-sizes of the same scale for all co-ordinates. 
We focus on the case where the target distribution concentrates on a linear hyperplane as $\epsilon \rightarrow 0$. 
Adapting the RWM steps in the direction of smaller scale $\epsilon$, we 
derive a diffusion limit for the RWM for the co-ordinates with $\mathcal{O}(1)$ scale. 
In contrast with all previous results on diffusion limits of MCMC algorithms, our limiting SDE will in general have non-additive noise, \textit{i.e.}, the diffusion term will not have constant coefficients. We also look at the case when the step-size is allowed to vary according to the local curvature and 
obtain a corresponding diffusion limit. We show that diffusion limits can be useful in certain situations for providing optimality criteria in these problems. 
We mention that in the case of a linear manifold, the user can easily adapt the different RWM step-sizes along the various directions to obtain a very effective algorithm of cost $\mathcal{O}(1)$.
The real practical interest lies with non-linear structures.
We have focused on a linear manifold to facilitate the development of theory. We also provide a conjecture for the diffusion limit in the practical case of a family of non-linear manifolds, but analytical derivations will require considerable future work. 
\item[ii.] When the dimension of the manifold where the target measure concentrates lags that of the entire state space only by one, 
we find the usual practice of adapting the step-size to control the acceptance probability is \emph{not}-optimal. In particular, our diffusion limits imply that a more efficient chain can be obtained by keeping a large step-size and allowing the acceptance probability to drop to zero. In this case we show that as $\epsilon \rightarrow 0$, the optimally-scaled RWM algorithm does \emph{not} converge weakly to a diffusion; rather it converges to a continuous-time Markov jump process.
\item[iii.] In contrast, when  the dimension of the manifold where the target measure concentrates lags that of the entire state space by at least $3$, we find that the diffusion regime is optimal. Intuitively, the cost of attempting large moves
(as measured by small acceptance probability) is just too large in this case. In the critical case when  the dimension of the manifold where the target measure concentrates lags that of the entire state space by exactly $2$, the cost of low acceptance probability is of the same order as the benefit of larger moves, so that
the optimal scaling strategy can vary depending on the specific form of the target density.
\item[iv.] As mentioned before, the most important practical benefit from deriving a diffusion limit is that it leads to an automated choice of the step-size: choose the step-size that maximizes the diffusion coefficient of the limiting SDE \cite{robe:97}. Doing so maximizes the speed measure of the limiting diffusion which in turn translates to maximizing the one-step expected jumping distance of the Markov chain at stationarity, or minimizing the first order auto-correlation for the chain. Since our limiting diffusions generally do not have a constant diffusion coefficient, this approach of maximizing the diffusion coefficient for choosing the step size is not valid any more. In general, there is no optimal value for the jump size of the RWM algorithm when applied to the target measures analyzed in this paper. We prove nevertheless that, if the jump size is allowed to be position dependent, optimality results can then be recovered. When the dimension of the coordinates with scale $\epsilon$ is large, our results yield a `local' $0.234$ rule that generalize the standard global $0.234$ of \cite{roberts1997weak}, \textit{i.e.}, it is asymptotically optimal to tune the local jump size so that the local acceptance rate equals $0.234$.
\item [v.] In general, practitioners should be aware that, when tuning the standard RWM algorithm, the general strategy consisting in choosing the largest jump size that yields an acceptance rate bounded away from a predetermined threshold (\textit{e.g.}, acceptance rate larger than $25\%$) is sometimes sub-optimal; indeed, our analysis shows that for `multiscale densities', optimal strategies can yield arbitrarily small acceptance rates.

\item [vi.] Technically, our proofs for the diffusion limits are different from the usual optimal scaling literature, as the scale parameter in our target measure need not be related to the dimension. Our main argument relies on the fact that by suitably rescaling the co-ordinates corresponding to the scale $\epsilon$,
the RWM algorithm mixes in these directions at a much faster rate than its $\mathcal{O}(1)$ counterpart. This fast mixing on the Markov chain in the smaller scale will lead to an `asymptotic decoupling' of the two scales which then will give the diffusion limit on the larger scales. The proof technique developed here could be applicable to other contexts and is thus of independent interest. 
\end{itemize}
%
%
%
The class of target measures we analyze in this paper should be considered as surrogates for distributions arising from real applied problems where components have  various  length scales that differ by orders of magnitude. In Statistics, there are many such examples - we provide here a few, including our motivating scenario from Bayesian inverse problems. 
\begin{itemize}
\item
In inverse problems, situations where the posterior distribution concentrates in the neighbourhood of a non-linear manifold abound, see for example \cite{radde2014convergence,cotter2009bayesian,knapik2011bayesian}. A typical situation involves the posterior distribution of a high-dimensional vector $x \in \RR^n$ observed through a low-dimensional, possibly noisy, non-linear measurement function $\Phi:\RR^n \to \RR^d$ with $d \ll n$; the data collected is distributed as $y \sim \Phi(x) + \textrm{(noise)}$. As the intensity $\sigma > 0$ of the noise decreases to zero, the posterior distribution $\pi(x \mid y)$ typically  concentrates  in a neighbourhood of thickness $\mathcal{\sigma}$ around  the manifold $\mathcal{M} \equiv \{x \in \RR^n :  \Phi(x) = y\}$.
\item
Suppose that we have a sequence of posterior distributions $\{ \pi _n (\theta )\}$, indexed by the data size $n$, where we can write
$\theta = (\theta _1, \theta _2)$ with $\theta _1$ representing the components of the posterior which are {\em identifiable} while $\theta _2$ remains unidentifiable so that its standard deviation remains $\mathcal{O}(1)$ for increasing $n$. For example, consider the posterior $\pi_n( \sigma,\mu)$ associated to $n \geq 1$ observations $y_i = X_{k/n}$ of the diffusion process $dX = \mu \, dt + \sigma \, dW$ on $[0,1]$; as $n \to \infty$, only the volatility coefficient is identifiable.
In the so-called \emph{regular} case, we might expect the marginal posterior of $\theta _1$ to contract at rate $n$ for instance. Therefore it is very natural to see this scale divergence as $n\to \infty $.
\item
Consider the context of maximum likelihood estimation (MLE) for non-regular likelihood functions,
For such problems, {\em super-efficiency} of MLEs is a well-known phenomenon, see for example \cite{Vaart}, in which the standard deviation of the
MLE shrinks to $0$ faster than $n^{-1/2}$ for data size $n$.
It may be that {\em super-efficiency} applies only to some of the model parameters, leading to the kind of
heterogeneously scaled target distributions we consider here. For concreteness, we consider the specific problem where the data is assumed to be drawn i.i.d from the model $X \sim \Exponential(\lambda )$ conditioned on $X \le \theta $. It is well-known that for this kind of example the posterior for $\lambda $ contracts at the regular rate of $n^{-1/2}$ while that for $\theta $ contracts at the much more rapid rate of $n^{-1}$. Thus again we get scale divergence as 
$n\to \infty $.
\end{itemize}
%
%
%
The paper is structured as follows. 
In Section \ref{sec:model} we describe the RWM algorithm for a density concentrating 
on a hyperplane with rate $\epsilon>0$. In Section \ref{sec:limit}  we state the regulatory conditions and
write the diffusion limits as $\epsilon\rightarrow 0$.
In Section \ref{sec:proof}, we prove the stated results.
In Section \ref{sec.vanishing.acceptance} we prove a limiting 
result involving a Markov jump-process, when the step-sizes are order $\mathcal{O}(1)$, 
so are not adapted to the size of the smallest co-ordinate. 
In Section \ref{sec:man} we describe in some detail a conjecture for generalizing our diffusion 
limit results in the context of non-linear manifolds.
Throughout, 
we provide comments about the implications of the theoretical results.
We finish with some conclusions and description of future work in Section \ref{sec:end}.
\section{Random-Walk Metropolis on Affine Manifold}
\label{sec:model}
As explained above, we prove analytical rigorous results in the case when the manifold is flat, i.e.\@ 
an affine subspace of the general state space. 
We   discuss later on in the paper extensions to more general manifold structures.
We model the affine scenario as follows. We
consider the target distribution $\pi_{\epsilon}:\RR^{n_x}\times \RR^{n_y}\mapsto\RR$, 
for integers $n_x,n_y\geq 1$ and $n=n_x+n_y$, 
with density with respect to the $n$-dimensional Lebesgue measure
\begin{align} \label{ridge.eq.target.density}
\pi_{\epsilon}(x,y) = \pi_X(x)\, \pi_{Y|X}(y|x) =  \tfrac{1}{\epsilon^{n_y}}\,
e^{A(x)} \, e^{B(x,y/\epsilon)} \ ,
\end{align}
for a `small' scalar $\epsilon>0$. The $x$-marginal has density $\pi_X(x)=e^{A(x)}$ independently of $\epsilon$. 
This is a scaled version of the probability 
distribution $\pi(x,y) \equiv \pi_1(x,y)= e^{A(x)}e^{B(x,y)}$. 
As $\epsilon \to 0$, the support of the sequence of distributions $\pi_{\epsilon}$ 
concentrates on the linear subspace $\m$:
\begin{equation}
\label{eq:man}
\m = \big\{ \, (x,y) \in \RR^{n_x} \times \RR^{n_y} \, : \, y = 0 \, \big\} 
\; \subset \; \RR^{n_x+n_y}\ , 
\end{equation}
of dimension $n_x$. 
Integer $n_x$ can sometimes be thought of as the dimension of the non-identifiability and the integer $n_y$ as the part of dimension
that is fully specified. For instance, 
in a small-noise or increasing-data context, one can consistently estimate $n_y$-parameters out of $n$.
Parameter $\epsilon$ can be thought of as the thickness
of the support of $\pi_\epsilon$ at a neighbourhood of $x\in \m$.

To obtain samples from $\pi_{\epsilon}$
we consider the RWM algorithm proposing moves
\begin{align}
\left(\begin{array}{c} 
X_{\epsilon}' \\ Y'_{\epsilon}  
\end{array}\right) = 
\left(\begin{array}{c} 
X_{\epsilon} \\ Y_{\epsilon}  
\end{array}\right) + 
\ell\,\, h(\epsilon)\,\left(\begin{array}{c} 
Z_x \\ Z_y  
\end{array}\right) 
 \label{ridge.eq:proposal}
\end{align}
for a scaling factor $h(\epsilon)$, tuning parameter $\ell>0$  
and noise $(Z_x,Z_y)\sim  \Normal(0,I_{n_x+n_y})$. 
The  factor $h(\epsilon)$ determines the scales of the jumps of the RWM algorithm and the tuning parameter $\ell$ allows to control the acceptance rate of the algorithm.
When the context is clear, we write simply $(X,Y)$ instead of $(X_{\epsilon}, Y_{\epsilon})$.
For the purpose of analysis, we introduce the rescaled coordinate $U_{\epsilon}$ and the associate rescaled proposal $U'_{\epsilon}$,
\begin{equation*}
U_{\epsilon} =  Y_{\epsilon} / \epsilon\ 
\quad \textrm{and} \quad
U'_{\epsilon} = Y_{\epsilon}'/\epsilon = U_{\epsilon} + \ell\, \tfrac{h(\epsilon)}{\epsilon}\,Z_y\ .
\end{equation*}
If $(X_{\epsilon},Y_{\epsilon}) \sim \pi_{\epsilon}$ then $(X_{\epsilon},U_{\epsilon})\sim \pi$.
To finish the description of the MCMC algorithm, we need to choose  
an accept-reject function $F$. To obtain detailed balance w.r.t.\@ $\pi_{\epsilon}$, one can choose any $(0,1]$-valued 
function $F$ satisfying the reversibility condition 
\begin{align} \label{ridge.eq.reversibility}
e^r F(-r)  = F(r)\ .
\end{align}
The usual Metropolis-Hastings accept/reject correction corresponds  to the choice    $F(r) = F_{\text{MH}}(r) = \min(1,e^r)$.
The move $(X,Y) \mapsto (X',Y')$, or equivalently $(X,U) \mapsto (X',U')$, 
is then accepted with probability
\begin{equation*}
F \circ \log \BK{ \frac{\pi_{\epsilon}(X',Y')}{\pi_{\epsilon}(X,Y)} }\ .
\end{equation*} 
We work with a general accept-reject function $F(r)$ as enforcing some 
differentiability condition removes several inessential technicalities 
from our proofs. 
Indeed, we assume in the remainder of this article the following regularity assumptions on $F$. 
\begin{assumptions}
\label{ass:F}
The accept/reject function $F$ is differentiable. $F$ and $F'$ are globally Lipschitz and bounded over the real line.
\end{assumptions}

\noindent
Assumption \ref{ass:F} could   be relaxed. For example, one could deal with the standard Metropolis-Hastings ratio; this would involve dealing with the discontinuity of the derivative at 0, which is only a technical distraction that does not bring new insights into the behaviour of the algorithm and make the proofs much more opaque.
Notice that any function $F$ satisfying 
the reversibility condition \eqref{ridge.eq.reversibility} is dominated by the Metropolis-Hasting 
function $F_{\text{MH}}$ in the sense that the inequality $F(r) \leq F_{\text{MH}}(r)$ holds 
for any $r \in \RR$. 
For the target density \eqref{ridge.eq.target.density} the acceptance probability can be written analytically as $a(X,U, h(\epsilon) \, Z_x, h(\epsilon) \, \epsilon^{-1} \, Z_y)$ where the function $a(\cdot, \cdot, \cdot, \cdot)$ reads
\begin{align}
\label{ridge.eq:accept}
a(x,u, \delta_x, \delta_u) = F\Big( A(x + \ell \, \delta_x ) - A(x) + B(x + \ell \, \delta_x, u + \ell \, \delta_u) - B(x,u) \Big)\ .
\end{align}
The above proposal and acceptance probability give rise to the 
MCMC trajectory $\{(X_{\epsilon,k},Y_{\epsilon,k})\}_{k\ge 0}$.
\subsection{Notation}
For $(\epsilon,x,u) \in \RR^+ \times \RR^{n_x} \times \RR^{n_y}$ we write  $\EE_{x,u}\,[\,\cdot \,]$ (or sometimes $\EE_{\epsilon, x,u}\,[\,\cdot \,]$) to denote the conditional expectation $\EE\,[\; \cdot \mid (X_{\epsilon,0},U_{\epsilon,0}) = (x,u)\,]$.  
For functions $e = e(x,u)$, we write $\EE_{\pi}\,[\,e(x,u)\,]$ 
to indicate that the expectation is considered under the standardised
law $(x,u)\sim\pi$.
We denote by $\|\cdot \|_p$ the $L_p$-norm, $p>0$.
We use the expression ``$e=e(\epsilon,x,u) = o_{L_1(\pi)}(1)$" to indicate that for function $e(\epsilon, x,u)$ we have $\lim_{\epsilon \to 0} 
\EE_{\pi}\,|\,e(x,u,\epsilon)\,|=0$.
We use the symbol $(\,\lesssim\,)$ to indicate inequalities that hold under multiplication 
with a constant which does not depend on critical parameters implied by the context.

%
%

\section{Diffusion Limit}
\label{sec:limit}
We study in this section the behaviour of RWM,
as $\epsilon\rightarrow 0$, for a scaling factor of the form
\begin{align} \label{eq.standard.scaling}
h(\epsilon) =  \epsilon \ .
\end{align}
Note that it is straightforward to prove that any other scaling factor 
$h(\epsilon)$ such that $h(\epsilon) / \epsilon \to \infty$ 
would lead to an algorithm with an acceptance rate that decreases 
to zero as $\epsilon \to 0$. This corresponds to an algorithm that 
proposes too large moves and such a context  is investigated in Section \ref{sec.vanishing.acceptance}. A scaling factor of the form 
\eqref{eq.standard.scaling} thus corresponds to the largest choice of 
jump-sizes that leads to an algorithm who acceptance rate does not 
degenerate to zero in the limit $\epsilon \to 0$.
To state one of our main results, we introduce the quantity 
\begin{align} \label{ridge.eq.mean.acceptance}
a_0(x, \ell) 
= \int_{\RR^{n_y}} \EE_{x,u} \Big[ F\Big( B(x,u + \ell\, \,Z_y )-B(x,u)\Big) \Big] \, e^{B(x,u)} \, du  \  ,
\end{align}
which corresponds to the  limiting average acceptance probability, as $\epsilon \to 0$, 
of the RWM algorithm when assuming stationarity for $U|X=x$, conditionally on a fixed position for the $x$-coordinate. 
One can easily verify (bounded convergence theorem) that
\begin{align}
a_0(x, \ell) = \lim_{\epsilon \to 0} \, \EE\, \big[\, a(X,U,\epsilon \, Z_x, Z_y) \mid  X=x\,\big]\ ,
\end{align}
assuming $(X,U)\sim \pi$.
We prove a diffusion limit for the trajectory of the $x$-coordinate, after considering 
a proper continuous time-scale for
its discrete-time trajectory $\{X_{\epsilon,k}\}_{k \geq 0}$, where $k$ denotes the number of MCMC iterations.  
To this end, we define
\begin{equation*}
c(\epsilon) = \epsilon^{-2}.
\end{equation*}
Our main result states that the sequence of accelerated, continuous-time, c\`adl\`ag processes
\begin{align}
\label{ridge.eq:xepsilon}
\wtilde{X}_{\epsilon,t} := X_{\epsilon, \floor{ t \cdot c(\epsilon) } }
\end{align}
converges weakly,  as $\epsilon \to 0$, to a non-trivial diffusion process; 
thus $c(\epsilon)$ corresponds to the `diffusive time-scale'.
\begin{remark}
For a given   positive density function $\pi: \RR^n \to (0,\infty)$, for $n\ge 1$, and a    scalar volatility function $\sigma: \RR^n \to (0,\infty)$ 
we introduce the function $\drift_{\pi, \sigma^2}: \RR^n \to \RR^n$ 
\begin{align}
\drift_{\pi, \sigma^2}: \quad x \mapsto \tfrac12 \, \nabla \sigma^2(x) + \tfrac12 \,\sigma^2(x)\,\nabla \log \pi(x) \ .
\end{align}
Consider the stochastic differential equation (SDE):
\begin{align} \label{eq.SDE.langevin}
d \overline{X} = \drift_{\pi, \sigma^2}(\overline{X}) \, dt  + \sigma(\overline{X}) \, dW
\end{align}
where $W_t$ denotes $n$-dimensional Wiener process.
Under standard regularity assumptions
the SDE \eqref{eq.SDE.langevin}
has a unique global solution and is reversible with respect to the law $\pi(dx)$. The case  $\sigma \equiv \textrm{const.}$ corresponds to the   Langevin diffusion  $d\overline{X} = \frac{\sigma^2}{2} \, \nabla \log \pi (\overline{X})\,  dt + \sigma \, dW$. 

\end{remark}
To obtain our scaling limit, we assume the following regularity conditions for functions $A$ and $B$ involved in the specification of the density $\pi_{\epsilon}$. 
\begin{assumptions} [Regularity Conditions on $\pi_{\epsilon}$]
 \label{ridge.assump.pi}
Functions $A:\RR^{n_x} \to \RR$ and $B: \RR^{n_x} \times \RR^{n_y} \to \RR$ are
differentiable and their derivatives 
are globally Lipschitz.
Moreover, we assume that the distribution $\pi \equiv \pi_1$ possesses finite second moments in the sense that
\begin{align} \label{ridge.eq.assumption.moment}
\EE_{\pi}\,[\,\|x\|^2 +\|u\|^2\,] < \infty\ .
\end{align}
\end{assumptions}
\noindent
Assumption \ref{ridge.assump.pi} is repeatedly used to control the error terms associated to the use of second order Taylor expansions; the second moment bound allows the use of the Cauchy-Schwarz inequality. Note that these assumptions could be relaxed in several directions at the cost of increasing complexity in the proofs; for example, one could assume only a polynomial growth on the derivatives
and higher moments for $\pi$.
 We also need  
the following assumption to control the behaviour of the diffusion limit identified below.

\begin{assumptions}
\label{ass:exist}
The function $x \mapsto \frac{1}{2}\,\nabla\BK{a_{0}(x,\ell)\log \pi(x)}$ is bounded on $\mathbb{R}^{n_x}$ and there is an exponent $\mu\in(0,1]$ such that 
\begin{align*}
\abs{ a_0(x,\ell)-a_0(x',\ell)} + 
\abs{ \nabla\BK{a_{0}(x,\ell) \times \log \pi(x)} - \nabla\BK{a_{0}(x',\ell) \times \log \pi(x')} }
\lesssim
\abs{x-x'}^{\mu}.
	\end{align*}
%
\end{assumptions}
\noindent
The main theorem of this section is the following. 
The proof is given in Section \ref{ridge.proof.main}.

%
%
\begin{theorem}
\label{ridge.th:main}
Let $T > 0$ be a fixed  time horizon. Assume 
that the RWM algorithm  is started in stationarity, 
$(X_{\epsilon,0}, Y_{\epsilon,0}) \sim \pi_{\epsilon}$, and
Assumptions~\ref{ass:F}-\ref{ass:exist} hold. Then,  as $\epsilon \to 0$ the sequence of  processes 
$\{ \wtilde{X}_{\epsilon,t} \}_{t \in [0,T]}$ defined via \eqref{ridge.eq:proposal}, \eqref{ridge.eq:xepsilon}
converges weakly in the 
Skorokhod space $\skorokhod([0,T], \RR^{n_x})$ to the diffusion process $\{\overline{X}_t\}_{t \in [0,T]}$
specified as the solution of the stochastic differential equation:
\begin{align} \label{ridge.eq.limit.diff}
d\overline{X} = \drift_{\pi_X, \sigma^2}(\overline{X}) \, dt  + \sigma(\overline{X})\,dW
\end{align}
for volatility function  $\sigma^2(x) \equiv \ell^2 \, a_0(x,\ell)$ and  initial position $\overline{X}_0 \sim \pi_X(x)=e^{A(x)}$.
\end{theorem}
%
%
Under Assumption \ref{ass:exist}, the limiting diffusion \eqref{ridge.eq.limit.diff} is well defined, does not explode in finite time, and has a unique strong solution. 
The assumption is as stated in \citep[ Ch. 4, Theorem 1.6]{ethier1986markov}. Furthermore, Assumption \ref{ass:exist} yields that for a smooth and compactly supported test function $\phi$ the function $x \mapsto \Gen \phi(x)$, with 
\begin{equation}
\label{eq:gene}
\Gen=\tfrac{\ell^2}{2} \nabla\{a_{0}(x,\ell) \log \pi(x)\}\cdot \nabla
+ \tfrac{\ell^2}{2}\,a_0(x,\ell)\,(\nabla \cdot \nabla)
\end{equation}
 the generator of the limiting diffusion \eqref{ridge.eq.limit.diff}, is $\mu$-Holderian.
The diffusive time scale $c(\epsilon) = \epsilon^{-2}$ implies that the algorithmic complexity 
of  RWM grows as $\OO(\epsilon^{-2})$ as the thickness $\epsilon$ approaches zero; see \cite{roberts2014complexity} for general results on the complexity analysis of MCMC algorithms through diffusion limits.
In the case where the function $(x,y) \mapsto B(x,y)$ does not depend 
on the $x$-coordinate, the limiting acceptance probability $a_0$ does not 
depend on the local position $x$, i.e.\@ $a_0(x,\ell) = a_0(\ell)$. In this case
the optimal value for the parameter $\ell$ is given by 
$\ell_* = \argmax \; \ell^2 \, a_0(\ell)$. In the general case however, 
the optimal choice of   $\ell$ will depend on the current $x$-position.

\subsection{Extending Theorem \ref{ridge.th:main}}

We can also adopt slightly 
more general proposals where the variance of the proposals is allowed to depend 
on the current position. That is, the tuning parameter $\ell = \ell(x) > 0$ is now allowed to depend 
on the $x$-coordinate:
\begin{align} \label{ridge.eq.general.proposal}
\left(\begin{array}{c} 
X_{\epsilon}' \\ Y'_{\epsilon}  
\end{array}\right) = 
\left(\begin{array}{c} 
X_{\epsilon} \\ Y_{\epsilon}  
\end{array}\right) + 
\ell(x) \, \epsilon \,\left(\begin{array}{c} 
Z_x \\ Z_y  
\end{array}\right) \ .
\end{align}
%
We state the required assumptions on the function $x \mapsto \ell(x)$ that allows the relevant diffusion limit results to hold. 
\begin{assumptions}[Regularity Assumptions on $x \mapsto \ell(x)$]
\label{ridge.assump.ell}
The function $\ell$ is positive, bounded away from zero and infinity. 
The first two derivatives of $\ell$ are   bounded.
\end{assumptions}
\noindent  
We choose to work in the limited setup of Assumption \ref{ridge.assump.ell} so that the proof of the next theorem is a straightforward adaptation of Theorem \ref{ridge.th:main}.
The accelerated version \eqref{ridge.eq:xepsilon} of the $x$-coordinate process again converges to a diffusion process.
\begin{theorem}
\label{ridge.th:main.general.proposal}
Let $T > 0$ be a fixed finite time horizon. Assume that 
Assumptions \ref{ridge.assump.pi}-\ref{ridge.assump.ell}  hold and  that 
the RWM algorithm  is started in stationarity.
As $\epsilon \to 0$, the sequence of processes 
$\{ \wtilde{X}_{\epsilon,t} \}_{t \in [0,T]}$ defined via (\ref{ridge.eq.general.proposal}), (\ref{ridge.eq:xepsilon}) converges weakly in the 
Skorokhod space $\skorokhod([0,T], \RR^{n_x})$ to the diffusion process $\{\overline{X}_t\}_{t \in [0,T]}$
specified as the solution of the stochastic differential equation
\begin{align} \label{ridge.eq.limit.diff.general}
d \overline{X} = \drift_{\pi_X, \sigma^2}(\overline{X}) \, dt  + \sigma(\overline{X})\,dW\  .
\end{align}
The local volatility function is $\sigma^2(x) \equiv \ell^2(x) \, a_0 (x,\ell(x))$. The initial distribution is $\overline{X}_0 \sim \pi_X$.
\end{theorem}
The only difference with Theorem \ref{ridge.th:main} is the form of the volatility function $\sigma$.
As before, the limiting distribution \eqref{ridge.eq.limit.diff.general} is reversible with respect to $\pi_X$ and 
the Dirichlet form reads
\begin{align}
\dirichlet(\phi) \equiv \tfrac{1}{2} \, \int_{\RR^{n_x}} \big\|\nabla \phi(x) \big\|^2 \, \ell^2(x) \,a_0\big(x,\ell(x) \big) \; \pi_X(dx)\ .
\end{align}
Since the parameter $\ell=\ell(x)$ is a function of the $x$-coordinate, the optimal 
choice $\ell_*(x)$ for the tuning parameter $\ell$ is
\begin{align}
\label{ridge.eq.pointwise.max}
\ell_\star(x) \equiv \argmax_{\ell > 0} \; \ell^2 \, a_0(x,\ell)\ . 
\end{align}
As described in \cite{robertsminimising},
the choice \eqref{ridge.eq.pointwise.max} maximises the $\esjd$, the spectral gap 
of the limiting diffusion \eqref{ridge.eq.limit.diff.general}
and the asymptotic variance of MCMC estimators.

%
%
\subsection{High-Dimensional Asymptotics and Local $0.234$ Rule}
In this section, we investigate the behaviour of the optimal jumping rule $x \mapsto \ell_\star(x)$ in the regime where the dimensionality of the identifiability is large i.e. $n_y \gg 1$. We adopt the setting where the dimension $n_x$ remains fixed while the dimension $n_y$ grows to infinity. For simplicity, we postulate a product form for the $y$-marginal density. That is, we investigate the sequence of densities $$\pi^{(n_y)}_\epsilon(x,y) = e^{A(x) + B^{(n_y)}(x,y/\epsilon)}\ , \quad x\in\RR^{n_x}\,, \,\,y\in \RR^{n_y}\ , $$ in the regime where $n_y \to \infty$ and $\epsilon \to 0$, with 
\begin{align*}
B^{(n_y)}(x,y) = \sum_{j=1}^{n_y} b(x,y_j)
\end{align*}
where $b: \RR^{n_x} \times \RR \mapsto \RR$ is such that for any $x \in \RR^{n_x}$ the function $y \mapsto e^{b(x,y)}$ is a proper density function on $\RR$. We denote by $\ell_\star^{(n_y)}(x_0)$ the optimal value of the jump size when the $x$-coordinate of the MCMC algorithm exploring $\pi_\epsilon^{(n_y)}$ stands at $x_0 \in \RR^{n_x}$. With the obvious notations, the previous sections show that $\ell_\star^{(n_y)}(x_0) = \argmax \, \big\{\, \ell^2 \, a^{(n_y)}_0(x_0,\ell) : \ell > 0 \big\}$ where
\begin{align*}
a^{(n_y)}_0(x_0,\ell)
=
\EE\,\Big[\, F\Big(\,\sum_{j=1}^{n_y} b(x_0, Y^{(x_0)}_j + \ell \, Z_j) - b(x_0, Y^{(x_0)}_j)\,\Big)\,\Big]
\end{align*}
for an i.i.d.\@ sequence $\{Y^{(x_0)}_j\}_{j \geq 1}$ of $\RR$-valued random variables with distribution $e^{b(x_0, y)} \, dy$ and an i.i.d.\@ sequence of standard Gaussian random variables $\{Z_j\}_{j \geq 0}$. We assume the following regularity conditions on the marginal density $y \mapsto e^{b(x_0, y)} \equiv \mu_{(x_0)}(y)$; this is the equivalent of the conditions (A1) and (A2) of \cite{robe:97}.
\begin{assumptions}
\label{ridge.assump.b}
The density $\mu_{(x_0)}(y)$ is twice differentiable, the function $y \mapsto \mu_{(x_0)}'(y)/\mu_{(x_0)}(y)$ is Lipschitz continuous and 
the random variable $Y^{(x_0)}$ with density $\mu_{(x_0)}$ is such that
\begin{align*}
\EE\,\Big[\,\big(\, \tfrac{\mu'_{(x_0)}}{\mu_{(x_0)}}(Y^{(x_0)}) \,\big)^8 \,\Big] < \infty
\quad \textrm{and} \quad
\EE\,\Big[\, \big(\,\tfrac{\mu''_{(x_0)}}{\mu_{(x_0)}}(Y^{(x_0)}) \,\big)^4 \,\Big] < \infty\ .
\end{align*}
\end{assumptions}

\noindent
A simple adaptation of Corollary $1.2$ of \cite{robe:97} yields that, under Assumption \ref{ridge.assump.b}, we have
\begin{align*}
\lim_{n_y \to \infty} \; a^{n_y}_0(x_0, \ell \, n_y^{-1/2})
=
\EE\,\big[\,F\,\big(\,\Normal(-I^2_{(x_0)}/2,I^2_{(x_0)})\,\big)\,\big] \equiv \overline{a}_0(\ell)
\quad \textrm{with} \quad
I^2_{(x_0)} = 
\EE\,\Big[\,\big(\tfrac{\mu'_{(x_0)}}{\mu_{(x_0)}}(Y^{(x_0}) \,\big)^2\,\Big]\ . \end{align*}
Corollary $1.2$ of \cite{robe:97} corresponds to the special case $F=F_{\text{MH}}$ where a closed form expression for $\overline{a}_0(\ell)$ exists; for concreteness, we will also consider the special case $F=F_{\text{MH}}$ although generalization to arbitrary accept/reject functions is straightforward. The function $\ell \mapsto \ell^2 \times \overline{a}_0(\ell)$ is maximized for $\ell_\star$ such that $\overline{a}_0(x_0, \ell_\star) = 0.234$ (to three decimal places) \cite{robe:97}; in other words, as $n_y \to \infty$, the optimal jumping rule $\ell_\star(x)$ can be chosen such that the local acceptance rate at $x_0 \in \RR^{n_x}$ equals approximately $0.234$. Indeed, the derivation of this rule-of-thumb relies on the product form assumption of the $y$-marginal and is only valid in the asymptotic regime $n_y \to \infty$. Nevertheless, this type of analysis has been shown to be empirically and theoretically \cite{breyer2004optimal,bedard2007weak,
beskos2009optimal,BedardRosenthal:2008,mattingly2012diffusion,PST12} relevant to more general  distribution structures.

%
%
\section{Proof of Diffusion Limit}
\label{sec:proof}
This section gives rigorous proofs of Theorems \ref{ridge.th:main} and \ref{ridge.th:main.general.proposal}.
The proof is based on \citep[Ch.~$4$, Theorem $8.2$]{ethier1986markov} which gives conditions under which 
the finite dimensional distributions of a sequence of
processes converge weakly to those of a Markov process.
\cite[Ch.~$8$, Corollary $8.6$]{ethier1986markov}
provides further conditions for this sequence 
of processes to be relatively compact in the appropriate topology 
and thus establish weak convergence of the
stochastic processes themselves. 

\subsection{Proof of Theorem \ref{ridge.th:main}}
\label{ridge.proof.main}
We first give a high-level description of the proof.
We introduce an intermediate time scale $\wtilde{c}(\epsilon) = \epsilon^{-\gamma}$ for an exponent $\gamma \in (0,\tfrac{1}{2})$,  
and consider the sub-sampled process $\{(S_{\epsilon,k}, V_{\epsilon,k})\}_{k \geq 0}$ defined as
\begin{align*}
\left\{
\begin{array}{ll}
(S_{\epsilon,0}, \, S_{\epsilon,1}, \, S_{\epsilon,2}, \, \ldots) 
\!\!&\!\!= (X_{\epsilon,0}, X_{\epsilon,\floor{\wtilde{c}(\epsilon)}}, 
\, X_{\epsilon,\floor{2 \, \wtilde{c}(\epsilon)}}, \, \ldots)\ , \\
(V_{\epsilon,0}, \, V_{\epsilon,1}, \, V_{\epsilon,2}, \, \ldots) 
\!\!&\!\!= (U_{\epsilon,0}, U_{\epsilon,\floor{\wtilde{c}(\epsilon)}}, 
\, U_{\epsilon,\floor{2 \, \wtilde{c}(\epsilon)}}, \, \ldots)\ .
\end{array}
\right.
\end{align*}
%
On this time-scale 
the $x$-process evolves slowly (i.e. $S_{\epsilon,k} \approx S_{\epsilon,k+1}$) 
whereas the $y$-process has the time to mix (i.e. $V_{\epsilon,k+1}$ 
is approximately independent from $V_{\epsilon,k}$). We define the continuous-time accelerated processes
\begin{align}
\label{ridge.eq:sepsilon}
\wtilde{S}_{\epsilon,t} = S_{\epsilon, \floor{ t \cdot c(\epsilon) / \wtilde{c}(\epsilon) }}\equiv 
S_{\epsilon, \floor{ t \cdot \epsilon^{\gamma-2} }}
\quad \textrm{and} \quad
\wtilde{V}_{\epsilon,t} \equiv 
V_{\epsilon, \floor{ t \cdot \epsilon^{\gamma-2} }}\ .
\end{align}
See Fig.\ref{fig:trace_i} for a graphical representation of all four main processes involved in our derivations. 
%

\begin{figure}[!h]
\vspace{-0.1cm}
\begin{center}
\includegraphics[width=0.9\textwidth, height=0.5\textheight]{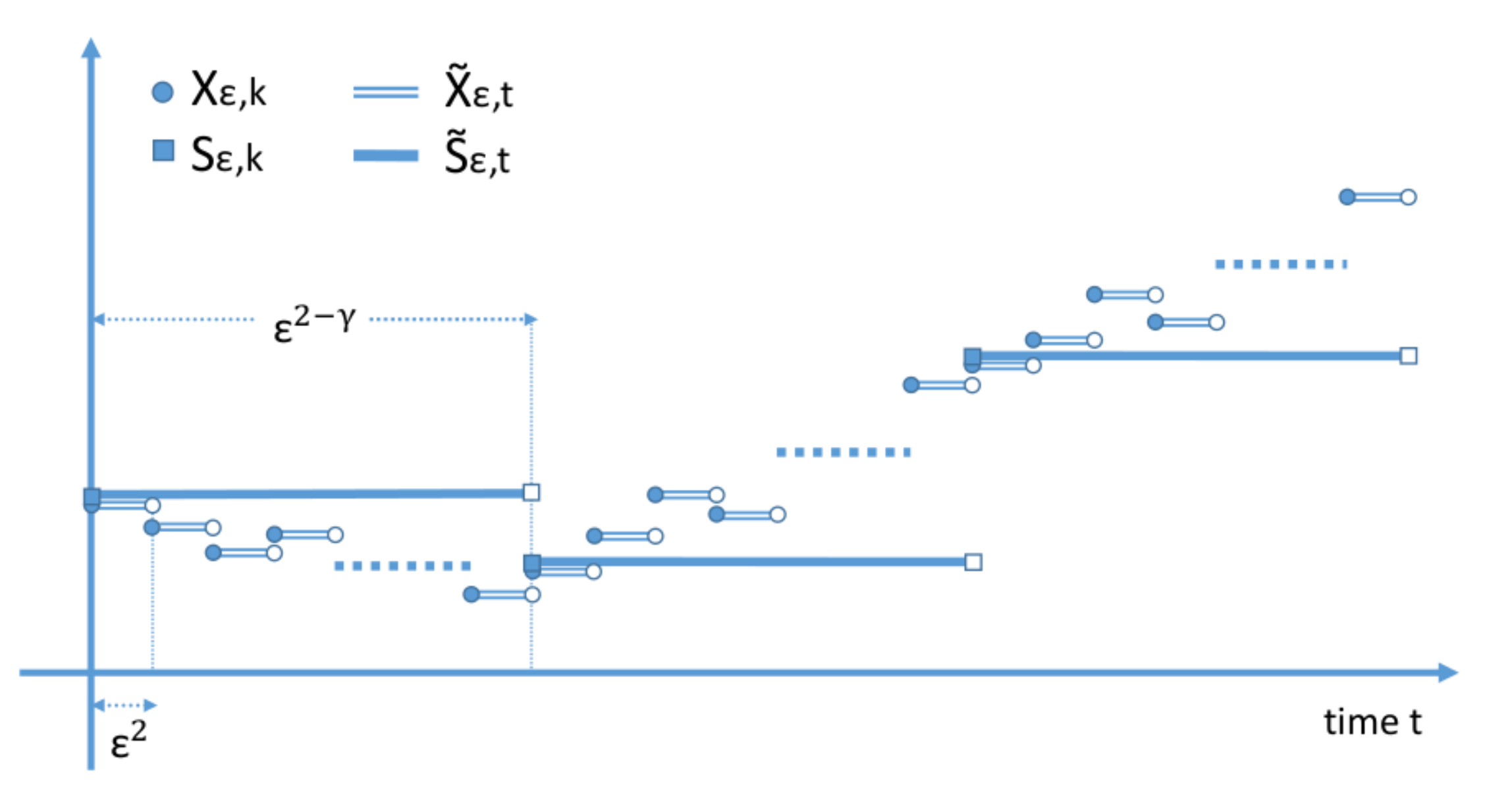} \qquad \qquad
\end{center}
\caption{The four processes involved in the diffusion limit result.
The discrete-time process $\{S_{\epsilon,k}\}_{k\ge 0}$ corresponds to a thinning (of frequency $1/\lfloor \epsilon^{-\gamma}\rfloor$) 
of $\{X_{\epsilon,k}\}_{k\ge 0}$; the first is illustrated with the filled rectangles, the latter with the filled circles.
Then, the continuous-time c\`adl\`ag process 
$\{\wtilde{S}_{\epsilon,t}\}_{t\ge 0}$ is illustrated with the filled lines
 and the $\{\wtilde{X}_{\epsilon,t}\}_{t\ge 0}$ one with the empty lines (the time instances indicated by the rectangles or the circles correspond to the jump times of the relevant processes).}
\label{fig:trace_i}
\end{figure}

The proof is divided into two steps. First, we show 
that the process $\wtilde{S}_{\epsilon}$  converges weakly in the Skorokhod topology to  the diffusion \eqref{ridge.eq.limit.diff}.
Then, we prove that the supremum norm 
\begin{align*}
\big\| \wtilde{S}_{\epsilon,\cdot} - \wtilde{X}_{\epsilon,\cdot} \big\|_{\infty, [0,T]}
\equiv \sup \Big\{ \; \big\| \wtilde{S}_{\epsilon,t} - \wtilde{X}_{\epsilon,t} \big\|\; : \; t \in [0,T] \; \Big\}
\end{align*}
converges to zero in probability, establishing 
the weak convergence of the sequence $\wtilde{X}_{\epsilon}$ itself towards the diffusion \eqref{ridge.eq.limit.diff}.
We define some quantities needed in the sequel. Recall the definition in (\ref{eq:gene}) 
of  the generator $\Gen$ 
of the limiting diffusion \eqref{ridge.eq.limit.diff}. Similarly,
we denote by $\Gen_\epsilon$ and $\wtilde{\Gen}_\epsilon$ the  generators of  the c\`adl\`ag processes  $\{\wtilde{X}_{\epsilon,t} \}_{t \geq 0}$ and    $\{\wtilde{S}_{\epsilon,t} \}_{t \geq 0}$ respectively. That is, let $\mathcal{C}$ be the space of smooth functions $\phi:\RR^{n_x}\mapsto \RR$ with compact support; for
any test function $\phi \in \mathcal{C}$ and  vectors $(x,u) \in \RR^{n_x} \times \RR^{n_y}$ we have 
\begin{align} \label{eq.generators}
\left\{
\begin{array}{ll}
\Gen_\epsilon \phi(x,u) 
\!\!&\!\!= \EE_{\epsilon,x,u} \big[ \phi(X_{\epsilon,1}) - \phi(x)\big] / \epsilon^2\ ,\\
\wtilde{\Gen}_{\epsilon} \phi(x,u)
\!\!&\!\! =\EE_{\epsilon,x,u}\big[
\phi(S_{\epsilon,1}) - \phi(x)\big] / \epsilon^{2-\gamma}\ .
\end{array}
\right.
\end{align}
Although the domain of $\phi$ is $\RR^{n_x}$, functions $\Gen_\epsilon \phi$ and $\wtilde{\Gen}_\epsilon \phi$ are defined on $\RR^{n_x} \times \RR^{n_y}$, as processes $\wtilde{X}$ and $\wtilde{S}$
are not Markovian with respect to their own filtration. 
The telescoping sum 
$
\phi(S_{\epsilon,1}) - \phi(S_{\epsilon,0}) =  
\phi(X_{\epsilon,1})
-
\phi(X_{\epsilon,0})
+
\cdots
+
\phi(X_{\epsilon,\floor{ \wtilde{c}(\epsilon)}})
-
\phi(X_{\epsilon,\floor{ \wtilde{c}(\epsilon) }-1})
$
yields that 
\begin{align}
\label{ridge.eq.telescop}
	\wtilde{\Gen}_\epsilon \phi(x,u) 
	&\equiv 
	\frac{1}{\epsilon^{-\gamma}} \,
	 \sum_{j=0}^{\floor{ \epsilon^{-\gamma}}-1}
	\EE_{\epsilon,x,u}\,\big[\, \Gen_{\epsilon} \phi (X_{\epsilon,j},U_{\epsilon,j}) \,\big]\ .
\end{align}
This identity is repeatedly used in the sequel. 
As a first step, 
we prove that for test function $\phi \in \mathcal C$,  generator 
$\Gen_{\epsilon} \phi (x,u)$ converges, in an appropriate sense, to the quantity $\A\phi(x,u)$ given by
\begin{align} 
\A\phi(x,u) 
= \ell^2 \, \Big\langle\,
\EE_{x,u}\,\big[\, F'(D B) &\nabla_x \big\{ A(x) + B(x,u + \ell Z_y) \big\} \,\big], \nabla \phi(x) \Big\rangle  \nonumber \\
&+ \tfrac{\ell^2}{2} \, \EE_{x,u}\,\big[\, F(D B)\,\big] \, \Delta \phi(x) \label{ridge.eq.def.A}\ , 
\end{align}
where for notational convenience we  have defined 
$D B = B(x,u+\ell Z_y) - B(x,u)$. Also, we have used the Laplacian notation 
$\Delta= \sum_{i=1}^{n_x}\partial^{2}_{x_i}$.
In general, $\A\phi(x,u)$ does not correspond to the generator of a Markov process. Note that if $\phi$ is smooth with compact support, under Assumptions \ref{ass:F}-\ref{ridge.assump.pi} we have that  $\abs{\A\phi(x,u)} \lesssim 1 + \norm{u}$.
We prove in Section \ref{ridge.sec.proof.limit.gen} the following result.
\begin{prop} \label{ridge.prop:limiting.generator}
Let Assumptions \ref{ass:F}-\ref{ridge.assump.pi} be satisfied 
and $\phi \in \mathcal C$ be a test function. 
Then, we have that $|\,\Gen_{\epsilon} \phi(x,u) - 
\A\phi(x,u) \,| \lesssim\epsilon\, \,(1+\|x\|+\|u\|)$,
thus 
 the following limit holds
\begin{align} \label{ridge.eq.generator.asymp}
\lim_{\epsilon \to 0} \; \EE_{\pi}\,\big|\,\Gen_{\epsilon} \phi(x,u) - 
\A\phi(x,u) \,\big|  \;=\; 0\ .
\end{align}
\end{prop}
\subsubsection{Weak convergence of $\wtilde{S}_\epsilon$ to the limiting diffusion \eqref{ridge.eq.limit.diff}}
\label{sec.conv.wtildeS}

We start by proving that the weak convergence of $\wtilde{S}_\epsilon$ towards the limiting diffusion \eqref{ridge.eq.limit.diff} reduces to studying the behaviour of the difference $\wtilde{\Gen}_{\epsilon} \phi(X,U) - \Gen \phi(X)$ between the approximate and limiting generators.

\begin{prop} \label{prop.L1.conv.generator}
Let Assumptions \ref{ass:F}-\ref{ass:exist} hold. 
If the following limit holds
\begin{align} \label{e.finite.dim.conv}
\lim_{\epsilon \to 0} \; \EE_{\pi}\,\big|\, \wtilde{\Gen}_{\epsilon} \phi(x,u) - \Gen \phi(x)\, \big| = 0\ ,
\end{align}
 then 
as $\epsilon \to 0$ the sequence of  processes 
$\{ \wtilde{S}_{\epsilon,t} \}_{t \in [0,T]}$ 
converges weakly in the 
Skorokhod topology $\skorokhod([0,T], \RR^{n_x})$ to the diffusion process \eqref{ridge.eq.limit.diff}.
\end{prop}
\begin{proof} For clarity, the proof is divided into two main steps. First, one proves that the finite dimensional distributions of $\wtilde{S}_\epsilon$ converge to those of the limiting diffusion. Secondly, one proves that the sequence $\wtilde{S}_\epsilon$ is relatively weakly compact in the appropriate topology.\\[0.3cm]
\noindent
{\it Convergence of the finite dimensional distributions of $\wtilde{S}_\epsilon$.}\\
%
We follow closely \citep[Ch.~$4$, Theorem $8.2$ and Corollary $8.5$]{ethier1986markov}
and apply those results for the  set of functions 
$E = \{(\phi, \Gen\phi):\,\phi\in \mathcal{C}  \}$.
Under Assumption \ref{ass:exist}, \citep[Ch. $8$, Theorem $1.6$]{ethier1986markov} gives 
that the closure of $E$ generates a Feller semigroup $\{T(t)\}$ (corresponding 
to the solution $X$ of the SDE) on the Banach space $L$ 
of continuous functions vanishing at infinity.
Thus, all conditions at the statement of \cite[Chapter $4$, Theorem $8.2$]{ethier1986markov}
are satisfied; it remains to prove part (e) of \cite[Ch.~$4$, Corollary $8.5$]{ethier1986markov}.
Given an arbitrary test function $\phi\in \mathcal{C}$, we 
set $f_{\epsilon}(x,u) = \phi(x)$ 
and $g_{\epsilon}(x,u)=\wtilde{\Gen}_{\epsilon}\phi(x,u)$.
We need to prove Equations $(8.8)$-$(8.9)$ and $(8.11)$ of \cite[Ch.~$4$, Theorem $8.2$]{ethier1986markov} (Equation $(8.10)$ is trivially satisfied); that is, one must show that
\begin{gather*}
\sup_{\epsilon>0}\,\sup_{t\le T}\,\,\, \EE\,|\,\phi(\wtilde{S}_{\epsilon,t}) | <\infty\ ; \\
\sup_{\epsilon>0}\,\sup_{t\le T}\,\,\, \EE\,|
\,\wtilde{\mathcal{\Gen}}_{\epsilon}\phi(\wtilde{S}_{\epsilon,t}, \wtilde{V}_{\epsilon,t})\,| <\infty\ ;\\
\lim_{\epsilon\rightarrow 0} \E{ \,\BK{ \wtilde{\Gen}_{\epsilon}\phi(\wtilde{S}_{\epsilon,t}, \wtilde{V}_{\epsilon,t}) - \Gen\varphi(\wtilde{S}_{\epsilon,t}) }\,\prod_{i=1}^{k}h_i(\wtilde{S}_{\epsilon,t_i})\, }  = 0\ ,
\end{gather*}
for any $k\ge 1$, any times $0\le t_1<\cdots <t_k\le t \le T$,
and any functions $h_i$ that can be assumed to be bounded.
Indeed, \citep[Ch. $8$, Theorem $1.6$]{ethier1986markov} shows that the above equations 
involving the generators and expectations at infinitesimally small increments from instance $t$
can imply convergence over finite times.  
The first requirement follows from the boundedness of $\phi$;
the second requirement is implied by the stationarity of
$(\wtilde{S}_{\epsilon,t}, \wtilde{V}_{\epsilon,t})$, 
Equation \eqref{e.finite.dim.conv} and the fact that $\EE_{\pi} \abs{ \Gen \phi(x) }<\infty$ 
(we have that $\sup_{x\in\RR^{n_x}}\abs{ \Gen \phi(x) }<\infty$ from the boundedness of the gradient of the drift function of the limiting diffusion 
on compact domains, since it is continuous from Assumption \ref{ass:exist}).
The third requirement is also implied from \eqref{e.finite.dim.conv}.
We have now proven the required convergence of the finite dimensional distributions 
of $\wtilde{S}_{\epsilon,t}$ to those of the solution of the limiting SDE (\ref{ridge.eq.limit.diff}).
\\[0.3cm]
\noindent
{\it Relative weak pre-compactness of $\wtilde{S}_\epsilon$}.\\
We follow \cite[Ch. $4$, Corollary $8.6$]{ethier1986markov}.
First, we remark that  the process $\wtilde{S}_\epsilon$ is started at stationarity
and the space  $\mathcal{C}\subset L$ of smooth functions with compact support is an
algebra that strongly separates points.
As noted in the proof of \cite[Ch.~$4$, Corollary $8.5$]{ethier1986markov}, the pair $(f_\epsilon(\wtilde{S}_{\epsilon,t}), g_\epsilon(\wtilde{S}_{\epsilon,t}, \wtilde{V}_{\epsilon,t}))$, with $f_{\epsilon}, g_{\epsilon}$ as defined above, in general does not belong to the approximate generator defined in Equation $8.6$ of \cite[Ch.~$4$]{ethier1986markov} and one needs to consider instead the pair 
$$
\BK{
f_\epsilon(\wtilde{S}_{\epsilon,t}) + (t-\floor{\epsilon^{\gamma-2} \, t}/\epsilon^{\gamma-2}) \, g_\epsilon(\wtilde{S}_{\epsilon,t}, \wtilde{V}_{\epsilon,t}), 
\;
g_\epsilon(\wtilde{S}_{\epsilon,t}, \wtilde{V}_{\epsilon,t})
}$$
 to account for the fact that the process $\BK{X_\epsilon, Y_\epsilon }$ is a discrete time Markov chain (note here the typo in Equation (8.28) of \cite[Ch.~$4$]{ethier1986markov}; the correct term involves the quantity $t - \floor{\alpha_n \, t} / \alpha_n$). Since $(t-\floor{\epsilon^{\gamma-2} \, t}/\epsilon^{\gamma-2}) < \epsilon^{2 - \gamma}$, to prove Equations $(8.33)$ and $(8.34)$ of \cite[Ch.~$4$, Corollary 8.6]{ethier1986markov} we must show that, for some exponent $p>1$, 
 some $\epsilon_0>0$ and for $(X,U)\sim \pi$, we have
\begin{align*}
\sup_{\epsilon\in(0,\epsilon_0)} \;
\big\| \wtilde{\Gen}_\epsilon \phi(X,U)\big\|_p < \infty
\qquad \textrm{and} \qquad
\lim_{\epsilon \to 0} \, 
\epsilon^{2-\gamma}\cdot \EE\,\big[\sup_{t\in[0,T]}|  \wtilde{\mathcal{\Gen}}_{\epsilon}\phi(\wtilde{S}_{\epsilon,t}, \wtilde{V}_{\epsilon,t})\,|\,\big] = 0\ .
\end{align*}
These will imply that the sequence $\wtilde{S}_\epsilon$ is relatively weakly pre-compact in the Skorokhod topology. Note that we have exploited the fact that the algorithm is started at stationarity.
For the first result, we  choose $1<p<2$. Then,  the telescoping expansion \eqref{ridge.eq.telescop} 
yields that
\begin{align*}
\big\| \wtilde{\Gen}_\epsilon \phi(X,U)\big\|_p
&\leq
\tfrac{1}{\epsilon^{-\gamma}} \sum_{j=0}^{\floor{ \epsilon^{-\gamma}}-1} \norm{ 
\Gen_{\epsilon} \phi (X_{\epsilon,j},U_{\epsilon,j}) - \A \phi (X_{\epsilon,j}, U_{\epsilon,j}) }_p
+
\tfrac{1}{\epsilon^{-\gamma}} \sum_{j=0}^{\floor{ \epsilon^{-\gamma}}-1} \norm{
\A \phi (X_{\epsilon,j}, U_{\epsilon,j}) }_p\\
&
\lesssim 
\,\norm{ \Gen_{\epsilon} \phi (X,U) - \A \phi (X,U) }_p
+
\norm{ \A \phi (X,U)}_p \ . 
\end{align*}
In the second line we have exploited the fact that the RWM chain is started at stationarity. 
The required result follows immediately from the given upper bound in 
Proposition \ref{ridge.prop:limiting.generator} and Assumption \ref{ridge.assump.pi}.
%
%
%
For the second result,  since the process $t\mapsto \wtilde{\mathcal{\Gen}}_{\epsilon}\phi(\wtilde{S}_{\epsilon,t}, \wtilde{V}_{\epsilon,t})$ makes at most $\lfloor T/\epsilon^{2-\gamma} \rfloor$ jumps on the interval $t \in [0,T]$, it suffices to show that each jump is $o(1)$ in $L_1$-norm; equivalently, due to the stationarity assumption, one needs to prove that the expectation
%
%
$
\| \,\wtilde{\mathcal{\Gen}}_{\epsilon}\phi( {S}_{\epsilon,1}, {V}_{\epsilon,1})-\wtilde{\mathcal{\Gen}}_{\epsilon}\phi( {S}_{\epsilon,0}, {V}_{\epsilon,0})\, \|_1
$
converges to zero as $\epsilon \to 0$.
Adding and subtracting $\Gen\phi(S_{\epsilon,1})-\Gen\phi(S_{\epsilon,0})$, Equation \eqref{e.finite.dim.conv} and the stationarity assumption yield that
\begin{align*}
\big\| \,\wtilde{\mathcal{\Gen}}_{\epsilon}\phi( {S}_{\epsilon,1}, {V}_{\epsilon,1})-\wtilde{\mathcal{\Gen}}_{\epsilon}\phi( {S}_{\epsilon,0}, {V}_{\epsilon,0}) \big\|_1
\leq 
o(1) + \|\, {\mathcal{\Gen}}\phi( {S}_{\epsilon,1})- {\mathcal{\Gen}} \phi( {S}_{\epsilon,0} )\,\|_1\ . 
\end{align*}
Under Assumption \ref{ass:exist}, for a smooth and compactly test function $\phi$, the function $x \mapsto \Gen \phi(x)$ is $\mu$-Holderian so that it suffices to show that $\norm{{S}_{\epsilon,1}-{S}_{\epsilon,0}}_1$ converges to zero as $\epsilon \to 0$; this is immediate since
\begin{align*}
\norm{ S_{\epsilon,1}-S_{\epsilon,0}}_1 
\;\lesssim \;
\epsilon  \, \sum_{j=1}^{\lfloor \epsilon^{-\gamma} \rfloor} \norm{Z_{x,j}}_1
\;\lesssim\;
\epsilon^{1-\gamma}
\end{align*}
and $\gamma$ was chosen inside the interval $(0,\tfrac{1}{2})$.

%
\end{proof}

\noindent
Proposition \ref{prop.L1.conv.generator} thus allows  to prove that the sequence $\{ \wtilde{S}_{\epsilon,t} \}_{t \in [0,T]}$ converges weakly in the Skorokhod space to the diffusion \eqref{ridge.eq.limit.diff} by establishing that  limit \eqref{e.finite.dim.conv} holds. Proving this result spans the remainder of this section.
%
%
\begin{prop} \label{prop:pinned}
Let Assumptions \ref{ass:F}-\ref{ridge.assump.pi} be satisfied 
and $\phi \in \mathcal C$ be a test function. 
Then, the following limit holds
\begin{align}
\lim_{\epsilon \to 0}  
\EE_{\pi}\,\Big|\,
\wtilde{\Gen}_{\epsilon}& \phi(x,u) 
-
\tfrac{1}{\epsilon^{-\gamma}} \, \sum_{j=0}^{\floor{\epsilon^{-\gamma}}-1} \EE_{\epsilon,x,u}\,\big[\,\A\phi(x,U_{\epsilon,j}) \,\big]\,\Big|  = 0\ . 
\label{eq:pinned}
\end{align}
\end{prop}
\begin{proof}
Under Assumptions \ref{ass:F}-\ref{ridge.assump.pi}  (we require that $F, F'$ are bounded and Lipschitz, $\nabla A, \nabla B$ are Lipschitz), for a smooth and compactly supported test function $\phi$, one can verify that 
\begin{equation}
\label{eq:lip}
|\A\phi(x,u)-\A\phi(\overline{x},u)| \lesssim (1+\|x\|+\|\overline{x}\|+\|u\|)
\norm{x-\overline{x}}
\end{equation}
 for any $\overline{x}, x \in \RR^{n_x}$ and $u \in \RR^{n_y}$. 
We now make use of the  
  telescoping equation 
\eqref{ridge.eq.telescop} to get that
\begin{align}
\EE_{\pi}\,\Big|\,
\wtilde{\Gen}_{\epsilon}& \phi(x,u) 
-
\tfrac{1}{\epsilon^{-\gamma}} \, \sum_{j=0}^{\floor{\epsilon^{-\gamma}}-1} \EE_{\epsilon,x,u}\,\big[\,\A\phi(x,U_{\epsilon,j}) \,\big]\,\Big|  \le 
\nonumber \\
&\le 
\EE_{\pi}\,\big|\,\Gen_{\epsilon} \phi(x,u) - 
\A\phi(x,u) \,\big| +
\tfrac{1}{\epsilon^{-\gamma}} \, \sum_{j=0}^{\floor{\epsilon^{-\gamma}}-1} \EE\,\big|\,\A\phi(X_{\epsilon,j},U_{\epsilon,j})- \A\phi(X_{\epsilon,0},U_{\epsilon,j})\,\big|\ . 
\label{eq:ant}
\end{align}
We consider the last term.  Using (\ref{eq:lip}) together with Cauchy-Schwartz, and the RWM upper bound
$\|X_{\epsilon,j}-X_{\epsilon,0}\|\lesssim \epsilon \sum_{k=1}^{j}\|Z_{x,k}\|$,
we obtain:
\begin{equation*}
\EE\,\big|\,\A\phi(X_{\epsilon,j},U_{\epsilon,j})- \A\phi(X_{\epsilon,0},U_{\epsilon,j})\,\big| \lesssim \epsilon 
\sum_{k=1}^{j} \|Z_{x,k} \|_2 = \mathcal{O}(j\,\epsilon)\ . 
\end{equation*}
Thus, the last term in (\ref{eq:ant}) is $\mathcal{O}(\epsilon^{1-\gamma})$
and vanishes in the limit  since $\gamma<1$.
The proof is  complete.
\end{proof}

To treat the average term in (\ref{eq:pinned}), 
we will we need to introduce a new Markov process 
coupled with the original one $\{(X_{\epsilon,k},U_{\epsilon,k})\}_{k\ge 1}$, 
but with the $x$-coordinate pinned at its initial position.
To this end, note that the Markov chain $\{X_{\epsilon,j}, U_{\epsilon,j}\}_{j \geq 0}$ can be described as follows. The initial position is defined as $(X_{\epsilon,0}, U_{\epsilon,0})=(X,U)$ for some variables $(X,U) \sim \pi$. For a sequence $\{\xi_j\}_{j \geq 0}$ of i.i.d random variables uniformly distributed on $(0,1)$ and a sequence $\{(Z_{x,j}, Z_{u,j})\}_{j \geq 0}$ of i.i.d random variables distributed as $\Normal(0,I_{n})$ we then have
\begin{align*}
\left(\begin{array}{c} 
X_{\epsilon,j+1}-X_{\epsilon,j} \\ 
U_{\epsilon,j+1}-U_{\epsilon,j} 
\end{array}\right) 
\;=\;
\ell \, \left(\begin{array}{c} 
\epsilon \, Z_{x,j} \\ 
Z_{y,j} 
\end{array}\right) 
\times 
\indic \,\BK{ \,\xi_{j} \leq 
a(X_{\epsilon,j},U_{\epsilon,j},\epsilon\, Z_{x,j},Z_{y,j}) \,}
\end{align*}
where the accept-reject function $a(\cdot, \cdot, \cdot, \cdot)$ is defined in   \eqref{ridge.eq:accept}. The new Markov chain $\{X^\star_j, U^\star_j\}_{j \geq 0}$ is defined as follows. For the \emph{same} random variables $(X,U)$ and $\{\xi_j\}_{j \geq 0}$ and $\{(Z_{x,j}, Z_{y,j})\}_{j \geq 0}$, we set $(X^\star_{0}, U^\star_{0})=(X,U)$ and 
\begin{align*}
\left(\begin{array}{c} 
X^\star_{j+1}-X^\star_{j} \\ 
U^\star_{j+1}-U^\star_{j} 
\end{array}\right) 
\;=\;
\ell \, \left(\begin{array}{c} 
0 \\ 
Z_{y,j} 
\end{array}\right) 
\times 
\indic \,\BK{ \,\xi_{j} \leq 
a(X,U^\star_{j},0,Z_{y,j}) \,}\  .
\end{align*}
Critically, the $x$-coordinate of the new process $\{X^\star_j, U^\star_j\}_{j \geq 0}$ remains still and the process does not depend on the parameter $\epsilon$. Also, conditionally on $X=x$, the process $\{U^\star_{j}\}_{j \geq 0}$ is simply a RWM Markov chain with target distribution on $\RR^{n_y}$ proportional to $u \mapsto \exp\BK{ B(x, u)}$. Thus, it readily follows from the Ergodic Theorem for Markov chains that for any smooth and compactly supported  test function $\phi$ we have
\begin{align} \label{eq.dif.A.gen}
\lim_{\epsilon \to 0} \; \EE_{\pi}\,\Big|\,  \tfrac{1}{\epsilon^{-\gamma}} \, \sum_{j=0}^{\floor{\epsilon^{-\gamma}}-1} \EE_{x,u}\,[\, \A\phi(x,U^\star_{j})\,] - \int_{u \in \RR^{n_y}} \A \phi(x,u) \, e^{B(x,u)}\, du\,\Big| = 0\ . 
\end{align}
Furthermore, a routine calculation, whose details 
can be found in Section \ref{ridge.sec.proof.lem.averaging}, 
gives the following result.
\begin{prop}
\label{pr:ident}
For any $x \in \RR^{n_x}$ and any $\varphi\in \mathcal{C}$ we have
\begin{align} \label{eq.identity.A.gen}
\int \A\phi(x,u) \, e^{B(x,u)} \, du
= \Gen \phi(x)\ .
\end{align}
\end{prop}
There is one result remaining to prove  weak convergence
of $\widetilde{S}_{\epsilon}$ to 
the limiting diffusion.
\begin{prop}
Let Assumptions \ref{ass:F}-\ref{ridge.assump.pi} be satisfied 
and $\phi \in \mathcal C$ be a test function. 
Then, the following limit holds
\begin{align}
\lim_{\epsilon \to 0}  
\EE_{\pi}\,\Big|\,
\tfrac{1}{\epsilon^{-\gamma}} \, \sum_{j=0}^{\floor{\epsilon^{-\gamma}}-1} \EE_{\epsilon,x,u}\,\big[\,\A\phi(x,U_{\epsilon,j}) \,\big]
- \tfrac{1}{\epsilon^{-\gamma}} \, \sum_{j=0}^{\floor{\epsilon^{-\gamma}}-1} \EE_{x,u}\,\big[\, \A\phi(x,U^\star_{j})\,\big] \,
\Big|  = 0\ . 
\label{eq:final}
\end{align}
\end{prop}

\begin{proof}
It suffices to establish that, for $X_{0,\epsilon}=X\sim \pi_X$,
\begin{align*} 
\lim_{\epsilon \to 0} \; \frac{1}{\epsilon^{-\gamma}} \, \sum_{j=0}^{\floor{\epsilon^{-\gamma}}-1}  \norm{\, \A\phi(X,U^\star_{j})   - \A\phi(X,U_{\epsilon,j})\,}_1 = 0\ .
\end{align*}
Under Assumptions \ref{ass:F}-\ref{ridge.assump.pi} and due the fact that $\varphi\in \mathcal{C}$, we have that  $\abs{\A\phi(x,u)} \lesssim 1 + \norm{u}$ so that $|\A\phi(X,U^\star_{j}) - \A\phi(X,U_{\epsilon,j})|$ is bounded by a constant multiple of $\indic(U^\star_{j} \neq U_{\epsilon,j}) \times (1 + \|U^\star_{j}\| + \|U_{\epsilon, j}\|)$. 
Recall that both chains are started from $(X,U)\sim \pi$. 
Also, from stationarity, we have that $U_j^*$ has the same law as $U$.
By the Cauchy-Schwarz inequality, since $\EE\,\|\,U\,\|^2 < \infty$
from Assumption \ref{ridge.assump.pi},  the conclusion follows once it is proved that
\begin{align*} 
\lim_{\epsilon \to 0} \; \frac{1}{\epsilon^{-\gamma}} \, \sum_{j=0}^{\floor{\epsilon^{-\gamma}}-1} \P{U^\star_{j} \neq U_{\epsilon,j}}^{1/2} = 0\ .
\end{align*}
The definition of the coupling between $(\{ X^\star_{j}, U^\star_{j}) \}_{j \geq 0}$ and $(\{ X_{\epsilon,j}, U_{\epsilon,j}) \}_{j \geq 0}$ shows that $U^\star_{j} = U_{\epsilon,j}$ if, and only if,
$
\indic(\, \xi_k \leq a(X_{\epsilon,k},U_{\epsilon,k}, \epsilon \, Z_{x,k}, Z_{y,k})\,) 
=
\indic(\, \xi_k \leq a(X,U_{\epsilon,k}, 0, Z_{y,k})\,)
$
for all $0\leq k \leq j-1$. It readily follows that
\begin{align*} 
\P{ U^\star_{j} \neq U_{\epsilon,j} }
&=
\EE \,\Big[\, 1-\prod_{k=0}^{j-1} \big(\,1-\big|\, a(X_{\epsilon,k},U_{\epsilon,k}, \epsilon \, Z_{x,k}, Z_{y,k})
- 
a(X,U_{\epsilon,k}, 0, Z_{y,k})\,\big|\,\big)\, \Big] \\
&\leq 
\sum_{k=0}^{j-1} \EE\, \big|\,
a(X_{\epsilon,k},U_{\epsilon,k}, \epsilon \, Z_{x,k}, Z_{y,k})
- 
a(X,U_{\epsilon,k}, 0, Z_{y,k})\,\big|
\end{align*}
where we have made use of the inequality $1-\prod_{k=0}^{j-1}(1-a_k) \le 
\sum_{k=0}^{j-1}a_k$, for sequences $a_k\in[0,1]$.
Under Assumptions \ref{ass:F}-\ref{ridge.assump.pi} we have that the difference 
$\big|
a(X_{\epsilon,k},U_{\epsilon,k}, \epsilon \, Z_{x,k}, Z_{y,k})
- 
a(X,U_{\epsilon,k}, 0, Z_{y,k})\,\big|$ is less than a constant multiple of 
$\epsilon \,(1+\|X\|+\|X_{\epsilon_k}\|+\|Z_{\epsilon,k}\|)\times\|Z_{x,k}\| +
(1+\|X\|+\|X_{\epsilon_k}\|+\|U_{\epsilon,k}\|+\|Z_{\epsilon,k}\|)\times \|X_{\epsilon,k}-X\|$.
Notice also that due to the RWM  chain we have $\|X_{\epsilon,k}-X_{\epsilon,0}\|\lesssim \epsilon \sum_{l=0}^{k-1}\|Z_{x,l}\|$.
Bringing everything together, we have shown that
$$\EE\, |\,
a(X_{\epsilon,k},U_{\epsilon,k}, \epsilon \, Z_{x,k}, Z_{u,k})
- 
a(X,U_{\epsilon,k}, 0, Z_{u,k})\,| \lesssim k \, \epsilon\ . $$ 
Therefore $\mathbb{P}( U^\star_{j} \neq U_{\epsilon,j} )\lesssim j^2 \, \epsilon$. This implies that $$\epsilon^{\gamma} \, \sum_{j=0}^{\floor{\epsilon^{-\gamma}}-1} \P{U^\star_{j} \neq U_{\epsilon,j}}^{1/2} \lesssim \epsilon^{1/2 - \gamma}\ .$$ Since $\gamma \in (0, \tfrac{1}{2})$, the conclusion follows. 

\end{proof}

This ends the proof of Equation \eqref{e.finite.dim.conv}, thus by Proposition \ref{prop.L1.conv.generator}, one can conclude that of the sequence of  processes 
$\{ \wtilde{S}_{\epsilon,t} \}_{t \in [0,T]}$ 
converges weakly in the 
Skorokhod space $\skorokhod([0,T], \RR^{n_x})$ to the diffusion process \eqref{ridge.eq.limit.diff}.\\

\subsubsection{Discrepancy Between $\wtilde{S}_\epsilon$ and $\wtilde{X}_\epsilon$}
\label{ridge.sec.conv.X}
We have proven in Section \ref{sec.conv.wtildeS} that 
the sequence of processes 
$\{\wtilde{S}_{\epsilon,t}\}_{t \in [0,T]}$ converges weakly in $\skorokhod([0,T], \RR^{n_x})$ 
to the limiting diffusion \eqref{ridge.eq.limit.diff}. This also implies that the sequence of 
processes $\{\wtilde{X}_{\epsilon,t}\}_{t \in [0,T]}$ converges towards the same limiting diffusion 
if one can establish that these two processes are close to each other in  the sense that
\begin{align} \label{ridge.eq.X.close.to.S}
\lim_{\epsilon \to 0} \; 
\EE\,\big[\, 
\sup \big\{\, \|\,\wtilde{X}_{\epsilon,t} - \wtilde{S}_{\epsilon,t} \,\| : t \in [0,T] \,\big\} \,\big] = \; 0\ .
\end{align}
Since the process $\wtilde{S}_{\epsilon,t}$ is obtained from the sequence $\{X_{\epsilon,j}\}$ by subsampling at rate $\tilde{c}(\epsilon) \equiv \epsilon^{-\gamma}$, the triangle inequality yields that the supremum in \eqref{ridge.eq.X.close.to.S} is less than a constant multiple of 
\begin{align} 
\epsilon \times \sup \,\Big\{\, 
\sum_{j=1}^{\floor{\tilde{c}(\epsilon)}} \norm{Z_{i,j}}
\; : \;
1 \leq i \leq \floor{ T \, c(\epsilon) / \tilde{c}(\epsilon) }
\,\Big\}
\end{align}
for independent centred and standard Gaussian random variables  $\{ Z_{i,j} \}_{i,j \geq 0}$ in $\RR^{n_x}$. 
To show that the above quantity converges to $0$ in expectation one can for instance work as follows. We define $$R_{i,\epsilon}:=\epsilon \sum_{j=1}^{\floor{\tilde{c}(\epsilon)}} \norm{Z_{i,j}}\ .$$
Then, for any $\alpha>0$, Markov's inequality gives $\mathbb{P}(R_{i,\epsilon}^3>\alpha^3)\le \EE\,[\,R_{i,\epsilon}^3\,]/\alpha^3\le C \epsilon^{3-3\gamma}/\alpha^3$ for a constant $C>0$. We also define 
$R_\epsilon:= \sup\{R_{i,\epsilon}: 1\leq i \leq \lfloor T/\epsilon^{2-\gamma} \rfloor\} $. Simple calculations give
\begin{align}
\EE\,[\,R_\epsilon\,] &= \int_{0}^{\infty} \mathbb{P}(R_{\epsilon}>\alpha)\,d\alpha
=  \int_{0}^{\infty} \Big[\,1 - \big\{\,1-\mathbb{P}(R_{1,\epsilon}>\alpha)\big\}^{\lfloor T/\epsilon^{2-\gamma} \rfloor}\,\Big]\,d\alpha  \nonumber \\
&\leq 
 \int_{0}^{\infty} \Big[\,1 - \big\{\,1-\tfrac{C}{\alpha^3}\,\epsilon^{3-3\gamma}\big\}^{\lfloor T/\epsilon^{2-\gamma} \rfloor}\,\Big]\,d\alpha  = 
  \int_{0}^{\infty} \Big[\,1 -\big[\,\big\{\,1-\tfrac{C}{\alpha^3}\,\epsilon^{3-3\gamma}\big\}^{\epsilon^{-3+3\gamma}}\big]^{\delta(\epsilon)}\,\Big]\,d\alpha 
  \label{eq:max}
\end{align}
with $\delta(\epsilon) =\epsilon^{3-3\gamma}\cdot \lfloor T/\epsilon^{2-\gamma} \rfloor$ vanishing in the limit since $\gamma\in(0,\tfrac{1}{2})$.
Now, for a $\delta>0$, we have that for big enough $\alpha$ the quantity  $\big\{\,1-\tfrac{C}{\alpha^3}\,\epsilon^{3-3\gamma}\big\}^{\epsilon^{-3+3\gamma}}$ is lower bounded by $e^{-(C/\alpha^3)\,(1+\delta)}$. 
Using this bound in (\ref{eq:max}) and then calling upon the dominated convergence theorem proves that $\EE\,[\,R_{\epsilon}\,]\rightarrow 0$ as required.
%

%
%
%
\subsection{Proof of Theorem \ref{ridge.th:main.general.proposal}}
\label{ridge.sec.proof.gen.prop}

The proof is entirely similar to the proof of Theorem \ref{ridge.th:main}. We only describe the modifications necessary 
to deal with this more general setting. We define the quantities $S_{\epsilon}, \wtilde{S}_{\epsilon}, \Gen_{\epsilon} \phi, \wtilde{\Gen}_{\epsilon} \phi$
the same way as in the proof of Theorem \ref{ridge.th:main}. The acceptance probability of the move $(X,U) \to (X + \ell \, \epsilon \, Z_x,U + \ell \, Z_y)$ 
reads
\begin{align*}
F \circ \log \, \BK{ \tfrac{\pi_{\epsilon}(X',U') \, p_{\epsilon}( (X',U') \to (X,U))}{\pi_{\epsilon}(X,U) \, p_{\epsilon}( (X,U) \to (X',U')) } }
\end{align*}
where $p_{\epsilon}[ (X,U) \to (X',U')]$ is the likelihood of the move $(X,U) \to (X'U')$. Proposition \ref{ridge.prop:limiting.generator} 
still holds but the limiting quantity $\A\phi(x,u) = \lim_{\epsilon \to 0} \, \Gen_{\epsilon} \phi(x,u)$ is now defined as
\begin{align*}
\A\phi(x,u)  =  \EE\,\big[\, F'(D B) \times\big\{ \ell^2(x)& \nabla_x ( A(x) + B(x,u + \ell Z_y))  + \nabla_x \ell^2(x) \big\} \,\big]
\cdot \nabla \phi(x) \\
&+ \tfrac{1}{2}\,\ell^2(x) \, \EE\,[\,F(D B)\,] \, \Delta \phi(x)\ .
\end{align*}
The proof uses a Taylor expansion of $\Gen_{\epsilon} \phi(x,u)$ with Assumptions \ref{ridge.assump.pi}-\ref{ridge.assump.ell} invoked to give a control on the error terms. 
Under boundedness assumptions on the function $x \mapsto \ell(x)$ the coupling used in the last part of the proof of Theorem \ref{ridge.th:main} is still valid and the rest of the proof then follows exactly the same lines as the proof of Theorem \ref{ridge.th:main}.

\section{Vanishing Acceptance Probability}
\label{sec.vanishing.acceptance}
We now consider the scenario where the target distribution $\pi_\epsilon$ is explored by a RWM algorithm that employs jump proposal of size $\OO(1)$; in other words and with the notations of the previous section: $h(\epsilon) =  1$. 
\subsection{Markov Jump Process Limit}
At an heuristic level, as $\epsilon \to 0$, a proposal $(X,Y) \mapsto (X',Y')$ is accepted only if $\norm{Y'}$ is of order $\epsilon$, which happens with probability $\OO(\epsilon^{n_y})$. In order to obtain a non-trivial limiting object, one thus needs to accelerate time by a factor $\epsilon^{n_y}$ and in this case the rescaled RWM trajectories converge, as $\epsilon \to 0$, to a Markov jump process limit.
In particular, we now consider the process $t \mapsto ( \wtilde{X}_{\epsilon,t}, \wtilde{U}_{\epsilon,t})$ defined as
\begin{equation}
\label{eq:new}
( \wtilde{X}_{\epsilon,t}, \wtilde{U}_{\epsilon,t})
\;=\;
(X_{\epsilon, \floor{ t \cdot \epsilon^{-n_y} }},
U_{\epsilon, \floor{ t\cdot \epsilon^{-n_y}} } )
\end{equation}
where, as in the previous section, we have used the rescaled coordinate $U_\epsilon \equiv Y_\epsilon / \epsilon$. For a smooth and compactly supported test function $\phi: \RR^{n_x} \times \RR^{n_y} \to \RR$, the generator of the process $t \mapsto (\wtilde{X}_{\epsilon,t}, \wtilde{U}_{\epsilon,t})$ reads
\begin{align*}
\GG_\epsilon &\phi(x,u)
=
\epsilon^{-n_y} \, \EE_{x,u}\,\big[\,\big(\,\phi(x+\ell\,Z_x, u + \ell\, \epsilon^{-1} Z_y) - \phi(x,u)\,\big) \, a(x,u, Z_x,  \, \epsilon^{-1} \, Z_y)\,\big]
\\
&=
\int_{\RR^{n}} \BK{\phi(\overline{x},\overline{u})-\phi(x,u)} \, 
Q(x,u,\overline{x},\overline{u}) \,
\exp\big\{-\epsilon^2 \, \norm{\overline{u}-u}^2/(2 \ell^2)\,\big\}\, d(\overline{x} , \overline{u})\ ,
\end{align*}
with $(X',U')=(x+\ell\,Z_x, u + \ell\, \epsilon^{-1} Z_y)$ where $(Z_x, Z_y)$ is a standard Gaussian random variable on $\RR^{n_x} \times \RR^{n_y}$ and the function $Q(\cdot, \cdot, \cdot, \cdot)$ reads
\begin{align*}
Q(x,u,\overline{x},\overline{u})
&=
\frac{F\big(\,A(\overline{x})-A(x)+B(\overline{x},\overline{u})-B(x,u)\,\big) \, 
\exp\big\{-\norm{\overline{x}-x}^2/(2\ell^2)\,\big\}}{(2\pi\ell^2)^{n/2}}\ .
\end{align*}
The next theorem shows that the sequence $t \mapsto (\wtilde{X}_{\epsilon,t}, \wtilde{U}_{\epsilon,t})$ converges to a Markov jump process $t \mapsto (\wtilde{X}_{t}, \wtilde{U}_{t})$ with transition kernel $K(x,u,\overline{x},\overline{u})=Q(x,u,\overline{x},\overline{u})/r(x,u)$ and jump rate   function $r(x,u) = \int Q(x,u,\overline{x},\overline{u}) \, d(\overline{x},\overline{u})$. Note that
$Q(x,u,\overline{x},\overline{u})\le \pi(\overline{x}, \overline{u})/\{\pi(x,u)(2\pi \ell^2)^{n/2}\}$, thus $r(x,u)<\infty$ and the limiting jump process is well-defined.
The generator $\GG$ of this jump process reads
\begin{align}
\label{eq:limJP}
\GG \phi(x,u)
=
\int_{\RR^{n}} \BK{\phi(\overline{x},\overline{u})-\phi(x,u)} \, 
Q(x,u,\overline{x},\overline{u}) \,
d(\overline{x} , \overline{u})\ .
\end{align}
Informally, the Markov process $t \mapsto (\wtilde{X}_{t}, \wtilde{U}_{t})$ can be described as follows; when found at $(x,u)$, the process waits an exponential time with parameter $r(x,u)$ before jumping to a new position $(\overline{x}, \overline{u})$ whose density is given by $K(x,u,\overline{x}, \overline{u})$. Using an approach similar to the one of the previous section, one can prove the following result. The homogenization argument of the previous section is not necessary since the $u$-coordinate does not need to be averaged out; the proof is much simpler.
We do not require strong differentiability conditions, 
the following assumption will suffice.
\begin{assumptions}
\label{ass:cont}
Functions $A:\RR^{n_x}\mapsto \RR$ and $B:\RR^{n_x+n_y}\mapsto \RR$ 
are continuous.
\end{assumptions}

\begin{theorem}
\label{th:jump}
Assume the that the process $(\wtilde{X}_{\epsilon}, \wtilde{U}_{\epsilon})$ is started at time $0$ from the equilibrium distribution $\pi$. 
Under Assumption \ref{ass:cont},
for a finite horizon $T>0$, the sequence of processes $t \mapsto (\wtilde{X}_{\epsilon,t}, \wtilde{U}_{\epsilon,t})$ converges weakly in the Skorokhod space $\skorokhod([0,T], \RR^{n})$  to the time-homogeneous   jump-process $t \mapsto (\wtilde{X}_t, \wtilde{U}_t)$ with generator $\GG$ in (\ref{eq:limJP}).
\end{theorem}

%
%
\begin{proof}
The proof is very similar to the one of Theorem \ref{ridge.th:main},
thus relies on showing convergence of the approximate generator $\GG_{\epsilon}$ to the limiting one $\GG$ in the sense of  
Proposition  \ref{prop.L1.conv.generator}. For the reader's convenience, we will thus only highlight  the differences with the proof of Theorem \ref{ridge.th:main}. As a first step, we prove that the finite dimensional distributions of $(\wtilde{X}_\epsilon, \wtilde{U}_\epsilon)$ converges to those of $(\wtilde{X}, \wtilde{U})$. The proof is then concluded by proving that the sequence $(\wtilde{X}_\epsilon, \wtilde{U}_\epsilon)$ is weakly pre-compact in the appropriate topology. \\[0.3cm]
\noindent
{\it Convergence of the finite dimensional distributions of $(\wtilde{X}_\epsilon, \wtilde{U}_\epsilon)$.}\\
As in the proof of Proposition \ref{prop.L1.conv.generator}, since the algorithm is assumed to start at stationarity, this reduces to proving that
\begin{align} \label{eq.l1.conv.jump.process}
\lim_{\epsilon \to 0} \;
\EE_{\pi}\,\big|\,\GG_\epsilon \phi(x,u) - \GG \phi(x,u) \,\big| \; = \; 0
\end{align}
for any smooth and compactly supported test function $\phi: \RR^{n_x} \times \RR^{n_y} \to \RR$. Recall also that we have the bound 
$F(r)\le F_{MH}(r)= 1\wedge e^{r}$, $r\in\mathbb{R}$.
We have
\begin{align*}
\EE_{\pi}\,\big|\,\GG_\epsilon \phi(x,u) - \GG \phi(x,u)\,\big|
&\leq 
\int \abs{D \phi} \, 
\BK{\pi(x,u) \wedge \pi(\bar{x}, \bar{u})} \, \BK{1-e^{-\frac{(\bar{u}-u)^2}{2 \ell^2}}} \, e^{-\frac{(\bar{x}-x)^2}{2 \ell^2}} \, \frac{d(x,u,\bar{x},\bar{u})}{(2 \pi\ell^2)^{n/2}}\\
&\lesssim
\int_{\overline{\Omega}} 
\min\BK{\pi(x,u), \pi(\bar{x}, \bar{u})} \, \BK{1-e^{-\frac{(\bar{u}-u)^2}{2 \ell^2}}} \, e^{-\frac{(\bar{x}-x)^2}{2 \ell^2}} \, d(x,u,\bar{x},\bar{u})
\end{align*}
where $D \phi \equiv \phi(\bar{x}, \bar{u}) - \phi(x,u)$; we have used the fact that since $\phi$ is smooth with compact support $\Omega$ the norm $\norm{D \phi}_{\infty}$ is finite and $D \phi$ is zero outside of $\overline{\Omega} = \BK{ \RR^n \times \Omega} \cup \BK{ \Omega \times \RR^n }$. Notice that 
\begin{align*}
\int_{\overline{\Omega}}  \min\BK{\pi(x,u), \pi(\bar{x}, \bar{u})} \, d(x,u,\bar{x},\bar{u}) \; \leq  \; 
2\int_{\Omega}dx du\,\int_{\mathbb{R}^{n}}\pi(x,u)dx du<\infty\  ,
\end{align*}
so Equation \eqref{eq.l1.conv.jump.process} follows from the dominated convergence theorem.\\

\noindent
{\it Relative weak pre-compactness of $(\wtilde{X}_\epsilon, \wtilde{U}_\epsilon)$}.\\
As in the proof of Proposition \ref{prop.L1.conv.generator},
we follow \cite[Ch. $4$, Corollary $8.6$]{ethier1986markov}; let us focus on proving that there exists an exponent $p>1$ such that
\begin{align} \label{eq.markov.jump.precompact}
\sup_{\epsilon>0} \;
\|\, \GG_\epsilon \phi(X,U)\|_p < \infty\ .
\end{align}
Since the density $\pi(x,y) = e^{A(x) + B(x,y)}$ is strictly positive and continuous, we have
\begin{align*}
\norm{ \GG_\epsilon \phi(X,U) }^p_p
&\leq \tfrac{1}{(2\pi \ell^2)^{np/2}}
\int \Big\{\,\int \, \abs{D \phi}
\min\big( \tfrac{1}{\pi(\overline{x},\overline{u})},\tfrac{1}{\pi(x,u)} \big) \, \pi(\overline{x},\overline{u}) \, d(\overline{x},\overline{u})\,\Big\}^p \, \pi(x,u) \, d(x,u) \\
&\leq
 \tfrac{1}{(2\pi \ell^2)^{np/2}} 2^p \, \norm{\phi}_\infty^p
\int_{\overline{\Omega}}
\min\big( \tfrac{1}{\pi(\overline{x},\overline{u})},\tfrac{1}{\pi(x,u)} \big)^{p} \, \pi(\overline{x},\overline{u}) \, \pi(x,u) \, d(x,u,\overline{x},\overline{u})\\[0.1cm]
&\leq \tfrac{1}{(2\pi \ell^2)^{np/2}} 
2^p \, \norm{\phi}_\infty^p \,
\sup \big\{\, \min\big( \tfrac{1}{\pi(\overline{x},\overline{u})},\tfrac{1}{\pi(x,u)} \big)^{p}
\, : \, (x,u,\overline{x},\overline{u}) \in \overline{\Omega}\, \big\} \\[0.1cm]
&=
 \tfrac{1}{(2\pi \ell^2)^{np/2}} 2^p \, \norm{\phi}_\infty^p 
\sup\curBK{ \pi(x,u)^{-p} \, : \, (x,u) \in \Omega }
 < \infty\ ,
\end{align*}
which proves Equation \eqref{eq.markov.jump.precompact}.
\end{proof}

\noindent
Since time has to be accelerate by a factor $\epsilon^{-n_y}$ in order to observe a non-trivial limit, the above theorem shows that the algorithmic complexity of RWM algorithm with jump proposal of order $\OO(1)$, when used to explore the distribution $\pi_\epsilon$, scales as $\OO(\epsilon^{-n_y})$.
Note nevertheless that it is not straightforward to optimize the free parameter $\ell>0$ since the limiting Markov jump processes obtained from different values of $\ell$ are generally not related by a simple linear change of time, as were the case for example in the original article \cite{roberts1997weak}.

\subsection{Jump process versus diffusion limit}
When using the RWM algorithm to explore $\pi_\epsilon$, the scaling limit Theorems \ref{ridge.th:main.general.proposal} and \ref{th:jump} reveal that in the case where the dimension of the weak identifiability equals one, i.e. $n_y=1$, it is asymptotically more efficient (as $\epsilon \to 0$) to use jump proposals of size $\OO(1)$, and thus using an algorithm with vanishing acceptance probability behaving like a Markov jump process, than adopting jump proposals of size $\OO(\epsilon)$ which leads to an acceptance rate bounded away from zero and one; indeed, it has been proved that in this case, jumps sizes of order $\OO(1)$ lead to an algorithm whose complexity scales as $\OO(\epsilon^{-1})$ whereas jump sizes of order $\OO(\epsilon)$ yield to a complexity that scales as $\OO(\epsilon^{-2})$. The standard rule-of-thumb that advocates tuning the mean acceptance probability of the random walk algorithm to a given optimal value $a_\star \in (0,1)$, for example equals to $a_\star = 0.234$ in the high-dimensional setting of \cite{roberts1997weak}, does not generally apply in the setting investigated in this article.\\

In the case where $n_y=2$, a similar argument shows that the approaches consisting in setting a jump size of $\OO(1)$ or $\OO(\epsilon)$ both yield to algorithms whose complexity scales as $\OO(\epsilon^{-2})$. In the case where the dimensionality of the weak identifiability is large, i.e. $n_y \geq 3$, it is preferable to adopt jump proposal sizes of order $\OO(\epsilon)$; the resulting algorithm scales as $\OO(\epsilon^{-2})$ whereas the choice of jump sizes $\OO(1)$ scales as $\OO(\epsilon^{-n_y})$.

\section{Towards a Diffusion Limit for General Manifolds}
\label{sec:man}
We give here a conjecture for a diffusion limit for the case 
of a non-linear manifold. An analytic
 proof is left for future work. 
The presentation will sidestep technicalities and will also serve 
to highlight the main building blocks of the earlier proof of Theorem \ref{ridge.th:main}.

We define an $n_x$-dimensional manifold by assuming existence 
of an invertible  global co-ordinate chart $r:\RR^{n_x}\mapsto \RR^{n_x+n_y}$,
so we have:
\begin{equation*}
\m = \big\{ \, v \in \RR^{n_x+n_y}\, :\, v = r(x)\ , \,x\in \RR^{n_x}\, \big\}\ .
\end{equation*}
We denote by $T_v \m$ the tangent space of $\m$ at $v=r(x)$. 
The plane $T_v \m$ is $n_x$-dimensional
with a canonical basis  comprised of the linearly independent 
vectors $\{ (\partial r /\partial x_i)(x)\}_{i=1}^{n_x}$.
The mapping $r$ gives rise to the metric tensor:
$$G(x)=\Big(\big\langle (\partial r/\partial{x_i})(x), 
(\partial r/\partial {x_j})(x) \big\rangle \Big)_{i,j=1}^{n_x}\in \RR^{n_x\times n_x}\ .$$
We will use the standard decomposition in terms of the tangent 
and its perpendicular normal space $N_{v}\m$ defined via  the linear system:
\begin{equation}
\label{eq:perp}
N_v \m = \{w\in \RR^{n_x+n_y}: \langle w, (\partial r/\partial x_i)(x)
\rangle=0\ , \,1\le i\le n_x \}\ .
\end{equation}
That is, we have  
%
$\RR^{n_x+n_y} = T_v\m\, \oplus\, N_v\m$. 
Let $(q_i(v))_{i=1}^{n_y}$ denote an orthonormal basis for $N_{v}\m$.
This could for instance be generated after applying Gram-Schmidt iteration
on the solutions of the linear system in \eqref{eq:perp}.
Some care is needed in the basis construction to ensure smoothness of $v\mapsto q_i(v)$,
$1\le i\le n_y$.
%
For $w\in N_v\m$ we denote by 
$Q_v w$ the ordered co-ordinates of $w$ w.r.t.\@ the basis $(q_i(v))_{i=1}^{n_y}$. 
We assume well-posedness of the projection
%
$\pr : \RR^{n_x+n_y} \mapsto \m$,
%
mapping each element of $\RR^{n_x+n_y}$ to its closest on $\m$ 
defined   as:
\begin{equation*}
\pr(w) = r\circ\big\{ 
 \arg\!\!\min_{x\in \RR^{n_x}}   \|r(x)-w\|^2 \big\} \  .
\end{equation*}
We will need the derivatives for the projection 
mapping. For a mapping $H:\RR^{k}\mapsto \RR^{l}$,
we denote $DH(x) = (\partial H_i/\partial x_j )_{1\le i\le l,1\le j\le k}$. 
We have:
\begin{align}
D\,(r^{-1}(\pr))(v) &= 
G^{-1}(x)\,\{Dr(x)\}^{\top}\ , \label{eq:T1}\\
D \pr(v) &= Dr(x)\,G^{-1}(x)\,\{Dr(x)\}^{\top}\ .  \label{eq:T2}
\end{align}
The above can be found by standard Taylor expansion of 
the distance metric $\|r(x)-w\|^2$ around its maximiser, akin 
to the procedure used for proving a CLT for the MLE, see for instance Chapter 18 of
\cite{ferg:96}.
%
For $w\in \RR^{n_x+n_y}$, a natural decomposition  to be used in this 
set-up is:
\begin{equation*}
w =  (x, y) \equiv \big(\,r^{-1}(\pr (w))\,,\, Q_v(w -\pr (w))\, \big)\ ;\quad x\in\RR^{n_x},\,y\in \RR^{n_y}
\ . 
\end{equation*}
The target distribution is assumed to be of the form:
\begin{align}
\pi_\epsilon(dw) = \pi_\epsilon(dx,dy) 
 = \tfrac{1}{\epsilon^{n_y}} e^{A(x) + B(x,y/\epsilon)} dx dy\ . 
\end{align}
To standarize we set $u=y/\epsilon$.
As in \eqref{ridge.eq:proposal}, we consider the standard RWM on $\RR^{n_x+n_y}$ 
with target $\pi_{\epsilon}(x,y)$ and proposal: 
\begin{equation*}
w' = w + \ell\,\epsilon\,Z\ , \quad Z\sim N(0,I_{n_x+n_y})\ . 
\end{equation*}
These dynamics give rise to the RWM trajectory $\{(X_{\epsilon,k},Y_{\epsilon,k})\}_{k\ge 0}$,
and the standardised trajectory $U_{\epsilon,k}=Y_{\epsilon,k}/\epsilon$.
We have that $w'=(x',y')$ with $x' = r^{-1}(\pr(w'))$.
A straightforward Taylor expansion using \eqref{eq:T1}-\eqref{eq:T2} will give:
%
%
%
\begin{align*}
x'&= x + \ell\,\epsilon\,J(x)\, Z +
\mathcal{O}(\epsilon^2)\ ;\\
u' &= u + \ell\,K(x)\,Z + \mathcal{O}(\epsilon)\ , 
\end{align*}
where we have defined:
\begin{align*}
J(x) &= G^{-1}(x)\,\{Dr(x)\}^{\top}\ , \\
K(x) &= Q(v)^{\top}\big(I-Dr(x)G^{-1}(x)\{Dr(x)\}^\top
\big)\  .
\end{align*}
for $Q(v)=[q_1(v),\ldots, q_{n_y}(v)]$.
%
%
%
We thus get that: 
%
%
\begin{align*}
\phi(x')-\phi(x) =  \big( \ell\,&\epsilon\,J(x) Z+ \mathcal{O}(\epsilon^2)\big)\,
(\nabla\phi(x))^{\top} \\ 
 &+ \tfrac{1}{2}\,\ell^2\,\epsilon^2\, \big\langle J(x)Z\,,\, \nabla^2\phi(x)\,J(x)Z\,\big\rangle 
 + \mathcal{O}(\epsilon^3)\ . 
\end{align*}
We now turn our attention to the acceptance probability term, and we have:
\begin{align*}
F\big(  A(x') -& A(x) + B(x',u')- B(x,u)   \big) =\\ &= F(B(x,u+\ell\,K(x)Z)-B(x,u)  )  + \mathcal{O}(\epsilon)\ . 
\end{align*}
Similarly to \eqref{ridge.eq.mean.acceptance}, 
we define the limiting average acceptance probability at position $x$:
\begin{align} \label{eq:average}
a_0(x, \ell) 
= \int_{\RR^{n_y}} \EE \Big[ F\Big( B(x,u + \ell\,K(x)Z )-B(x,u)\Big) \Big] \, e^{B(x,u)} \, du  \  .
\end{align}
Following the crux of the analytical proof for the case of affine manifold,
we start by looking at the one-step generator:
\begin{align} \label{eq:1a}
\Gen_{\epsilon} \phi (x,u)  &= \EE\,\Big[\, 
\frac{\phi(X_{\epsilon,1})-\phi(X_{\epsilon,0})}{\epsilon^{2}} 
\, \big| X_{\epsilon,0}=x, U_{\epsilon,0}=u\Big] \nonumber \\
& =   \EE_{x,u}\Big[\, 
\frac{\phi(x')-\phi(x)}{\epsilon^{2}}\cdot F (  A(x') -  A(x) + B(x',u')- B(x,u) )   \,\Big] 
\nonumber
\\[0.2cm]
&= \EE_{x,u}\,\big[\,\big( \ell\, \epsilon^{-1}\,J(x) Z+ \mathcal{O}(1)\big)\,
 F(B(x,u+\ell\,K(x)Z)-B(x,u)  )\,\big]\cdot 
(\nabla \phi(x))^{\top} \nonumber  \\[0.1cm] & \qquad \quad + \EE_{x,u}\,\big[\,\tfrac{1}{2}\ell^2  \big\langle J(x)Z\,,\, \nabla^2\phi(x)\,
J(x)Z\,\big\rangle\cdot F(B(x,u+\ell\,K(x)\,Z)-B(x,u)  )\,\big] \nonumber \\[0.1cm]
& \qquad \quad\qquad \quad \qquad \quad + \mathcal{O}(\epsilon)\ .  
\end{align}
Notice that due to the orthogonality of $T_v\m, N_v \m$, 
we have that:
\begin{equation*}
K(x)^{\top}J(x)=0\ ,
\end{equation*}
which implies the independence $J(x)Z\perp K(x)Z$. 
Thus, continuing from \eqref{eq:1a} we have that:
\begin{align}
\Gen_{\epsilon} \phi (x,u) &= \langle \mathcal{O}(1), 
\nabla \phi(x)\rangle \nonumber  \nonumber \\[0.1cm] & \qquad \quad + \EE_{x}\,\big[\,\tfrac{1}{2}\ell^2 
 \big\langle J(x)Z\,, \nabla^2\phi(x)\,
J(x)Z\,\big\rangle\,\big] \cdot \EE_{x,u}\,\big[\,  F(B(x,u+\ell\,K(x)\,Z)-B(x,u)  ) \,\big] \nonumber \\
& \qquad \quad\qquad \quad \qquad \quad + \mathcal{O}(\epsilon)\ . \label{eq:a2}
\end{align}
Following closely the affine case, 
we consider  the sped-up process $\wtilde{S}_{\epsilon,t} = S_{\epsilon, \floor{ t \cdot c(\epsilon) / \wtilde{c}(\epsilon) }}\equiv 
S_{\epsilon, \floor{ t \cdot \epsilon^{\gamma-2} }}$
for the sub-sampled trajectory $S_{\epsilon,k},V_{\epsilon,k}$. 
Recall that the idea here is the $u$-trajectory will have enough time to mix during the sub-sampled times, 
whereas the $x$-trajectory will still make local moves and provide a diffusion limit. 
Thus, we make the following conjecture for the generator $\wtilde{\Gen}_\epsilon$ of the
process $\wtilde{S}_{\epsilon}$:
\begin{align*}
	\wtilde{\Gen}_\epsilon \phi(x,u) 
	&=
	\tfrac{1}{\epsilon^{-\gamma}} \,
	\EE\Big[\,\sum_{j=0}^{\floor{ \epsilon^{-\gamma}}-1}
	\Gen_{\epsilon} \phi (X_{\epsilon,j},U_{\epsilon,j})\mid X_{\epsilon,0}=x, U_{\epsilon,0}=u \, \Big] 
	 = \EE_{x}\big[\,\Gen_{\epsilon}\phi(x,  u)\,\big] + 
o(1)
	\\
  &=  \langle \mathcal{O}(1),\nabla \phi(x)\rangle   + \EE_{x}\,\big[\,\tfrac{1}{2}\ell^2 
 \big\langle J(x)Z\,, \nabla^2\phi(x)\,
J(x)Z\,\big\rangle\,\big] \cdot a_0(x,\ell) + o(1)\  .
\end{align*}
%
%
%
It is easy to check that $\{Dr(x)\}^{\top} Dr(x) = G(x)$. Thus, the above 
expression, and in particular the quantity involving $\nabla^{2}\phi(x)$, 
suggests a diffusion limit with diffusion coefficient $\sigma$ such that:
 $$(\sigma\sigma^{\top})(x)=G(x)^{-1}a_0(x,\ell)\,\ell^2\ ,$$
 with a corresponding expression for the limiting SDE:
%
\begin{align} 
\label{eq:conj}
d\overline{X}_t = \text{drift} \big(\pi_X, \sigma\sigma^{\top})(\overline{X}_t) \, dt  + \sigma(\overline{X}_t)\,dW_t
\end{align}
%
for  
 $D_0 \sim \exp\{A(x)\}$ and $\text{drift} \big(\pi_X, \sigma\sigma^{\top}) = \tfrac{1}{2}\sigma \sigma^{\top}\nabla + \tfrac{1}{2}(\sigma \sigma^{\top})\nabla \log \pi_X$, with $\pi_X(x)=\exp\{A(x)\}$.
The expression we obtained for the diffusion coefficient 
is the same as the one for the Langevin SDE 
on a manifold obtained in \cite{girolami2011riemann,livingstone2014information} with the addition of the 
average acceptance probability term $a_0(x,\ell)\ell^2$.
\section{Conclusions/Future work}
\label{sec:end}
As far as we are aware, ours is the first attempt  toward analytically studying the behaviour and complexity of MCMC algorithms when applied to target densities with a multi-scale structure. We acknowledge here that the practical advice stemming out of our results are probably 
not as strong as in the case of diffusion limits in high-dimensions. 
Still, we believe that our analysis provides inroads for the investigation of MCMC algorithmic performance in a different direction from the one followed so far in the literature. Our work opens up a number of avenues for future work in this area. We highlight a few of these below.
\begin{itemize}

\item In many practical problems the limiting manifold will be non-linear and the directions of small size can vary in different parts of the state space 
and thus one cannot predetermine narrow directions and adjust the step-sizes. The conjectures about diffusions limits on manifolds in Section \ref{sec:man} thus have immediate impact in applications but might require substantial amount of work to proved in full generality.
\item
 In a wider perspective, we believe that the results 
in the paper open new directions also for the study of MCMC algorithms that better exploit the manifold 
structure of the support of the target distribution; this direction also connects with recent advances 
in the development of Riemannian MALA and HMC methods as in, \textit{e.g.}, \cite{girolami2011riemann,livingstone2014information}. 
The set-up in our paper is a bit more involved as the manifold can be of smaller 
dimension  that the general space (in the above works it is of the same dimension).
To be more explicit, following the notation of Section \ref{sec:man}, 
it would be of interest, for instance,  to study RWM with location-specific step-sizes, say of the form:
\begin{align*}
w' &= w + \ell\,\big(  \partial r/\partial x_1,\, \partial r/\partial x_2,\, \ldots,\,  
\partial r/\partial x_{n_x} \mid q_1,\, q_2,\, \ldots,\, q_{n_y} \big)\,
\left( \begin{array}{c} 
h (\epsilon) Z_x \\
\epsilon Z_y
\end{array}
\right)\\
&= w + \ell\, h(\epsilon)\, Dr(x)\,Z_x + \ell\, \epsilon\, Q(v)\, Z_y \ ,  
\end{align*}
 so that the method moves along the tangent space $T_{v}\m$ with a step of size $h(\epsilon)$
 and along the perpedicular normal space with step of size $\epsilon$ (as this is the 
 size of the probability mass around $\m$). Of interest here would be the optimal selection 
 for $h(\epsilon)$ and the specification of the computational cost for the method.
 For example, consider the case when the manifold $\m$ corresponds to a circle of radius 1 in 
 $\RR^{2}$. Then, it appears that controlling the acceptance probability would require 
 $h(\epsilon)=\sqrt{\epsilon}$, as this is the order of the size of the string of the circle
 which is perpendicular to a radius at position of distance $\epsilon$ from the circle surface. 
Such questions could be investigated for general manifold structures.
The above of course corresponds to an idealized algorithm, when the method 
uses explicit information about the tangent and normal space. 
In practice, it will be of interest to investigate also practical recent algorithms using for instance information about the curvature of the log-target distribution or the Hessian to scale the step-sizes in the different directions, and contrast their effect with the idealised scenario.
\item It is of interest to provide connections with locally adaptive methods 
currently looked at in the literature. 
Further exploration of such advanced methods in a similar manifold setting may provide analytical results that should be contrasted with the ones here and illustrate their superiority. 
\item Finally, all of our results assume that the Markov chain is at stationarity. There is a parallel literature on the scaling and the behavior of MCMC algorithms in the transient phase. The limiting process in the transient phase is usually an Ordinary Differential Equation instead of an SDE. It is natural next step to obtain analyze the transient phase of our algorithm.
 \end{itemize}
\section*{Acknowledgements}
AB acknowledges supports from a Leverhulme Trust Prize and AT from a Start-Up Grant from the National University of Singapore. NSP is partially supported by an ONR grant.
We thank Sam Livingstone at UCL for several useful 
conservations on the topic.

\vspace{0.5cm}

\appendix

\section{Proofs}
\label{app.A}

\subsection{Proof of Proposition \ref{ridge.prop:limiting.generator}} \label{ridge.sec.proof.limit.gen}
Recall the definition of the one-step generator $\Gen_{\epsilon}$   in   \eqref{eq.generators}.
Using the  notation $\mathsf{v} = \nabla \phi(x) \in \RR^{n_x}$, $\mathsf{S} = \nabla^2 \phi(x) \in \RR^{n_x \times n_x}$, a 2nd order Taylor expansion yields
\begin{align*}
\Gen_\epsilon \phi(x)
&=
\epsilon^{-2} \, \EE_{\epsilon,x,u} \,\big[\,\BK{\phi(x+\ell \, \epsilon \, Z_x) - \phi(x)} \times a(x,u,\epsilon \, Z_x,  Z_y)\,\big]\\
&=
\ell \, \epsilon^{-1} \, 
\EE_{\epsilon,x,u}\,\big[\,\bra{\mathsf{v}, Z_x} \times a(x,u,\epsilon \, Z_x, Z_y)\,\big] \\
&\quad +
\tfrac{1}{2}\,\ell^2 \, 
\EE_{\epsilon,x,u}\,\big[\, \bra{Z_x, \mathsf{S} \, Z_x} \times a(x,u,\epsilon \,Z_x, \, Z_y)\,\big] + o_{L_1(\pi)}(1)\ , 
\end{align*}
where the remainder term has been identified as $o_{L_1(\pi)}(1)$ for $\epsilon\rightarrow 0$
as its absolute value is upper bounded by $C\,\epsilon \,\EE\|Z_x\|^3$,
for a constant $C>0$,
due to $\varphi$ being smooth and of compact support.
Thus, to prove the stted limiting result, it suffices to prove that the following two identities hold,
\begin{align}
\epsilon^{-1} \, 
\EE_{\epsilon,x,u}\,
\big[\,\bra{\mathsf{v}, Z_x} \times a(x,u,\epsilon  \, Z_x,  Z_y)\,\big] 
&=\ell\,
\EE_{x,u}\,\big[\,F'(D B) \bra{ \mathsf{v},\nabla A(x) + \nabla_x B(x,u + \ell Z_y) }\,\big]  \label{eq:abc} \\  & \qquad \qquad \qquad +  o_{L_2(\pi)}(1)\  ,\nonumber \\[0.3cm]
\EE_{\epsilon,x,u}\,\big[\, \bra{Z_x, \mathsf{S} \, Z_x} \times a(x,u,\epsilon \,  Z_x,  Z_y)\,\big]
&=
\trace(\mathsf{S}) \times \EE_{x,u}\,[\,F(D B)\,] +  o_{L_2(\pi)}(1)\  . 
\label{eq:alex}
\end{align}
Recall the shorthand notation $D B= B(x,u + \ell \, Z_y)-B(x,u)$.
Expression \eqref{ridge.eq:accept} for the acceptance probability function $a(\cdot, \cdot, \cdot, \cdot)$ together with regularity Assumption \ref{ridge.assump.pi} on functions $A$ and $B$ yield
\begin{align}
a(x,u,\epsilon \, Z_x,  Z_y)
&=
F(D B) + \ell \, \epsilon \,F'(D B)\, \bra{\nabla A(x) + \nabla_x B(x,u + \ell \, Z_y), Z_x} 
+ \epsilon \times o_{L_1(\pi)}(1)\ .\label{eq:acca}
\end{align}
The remainder term has been identified as  
$\epsilon \times o_{L_1(\pi)}(1)$ as it is upper bounded in absolute value 
by $C\,\epsilon^2\,(1+\|x\|+\|u\|+\|Z_x\|+\|Z_y\|)\times \|Z_x\|$, for a constant $C>0$,
 due to $F'$ being bounded and Lipschitz, 
and $\nabla A, \nabla B$ being Lipschitz; also, $\pi$ has finite absolute first moments.
Using this expression gives 
 that
the quantity 
$
\epsilon^{-1} \, 
\EE_{\epsilon,x,u} 
\sqBK{ \bra{\mathsf{v}, Z_x} \times a(x,u,\epsilon \,Z_x,  Z_y)}
$
equals
\begin{align*}
\ell\,
\EE_{x,u}\,\big[\,\bra{\mathsf{v}, Z_x} \times F'(D B)\times 
\bra{\nabla A(x) + \nabla_x B(x,u + \ell \, Z_y), Z_x}\,\big] + o_{L_1(\pi)}(1) \ .
\end{align*}
which gives immediately (\ref{eq:abc}) after taking the expectation over $Z_x$.
Also,  (\ref{eq:alex}) follows immediately from (\ref{eq:acca}). 
Finally, the stated upper bound in the proposition follows immediately from the explicit 
upper bounds given above for the remainder terms.

%

%
%
\subsection{Proof of Proposition \ref{pr:ident}}
\label{ridge.sec.proof.lem.averaging}
To keep this exposition as simple as possible, we suppose that $\ell = 1$ 
and $n_x=n_y=1$. The multi-dimensional case is entirely similar.
The proof of \eqref{eq.identity.A.gen} consists in verifying that 
for all $x \in \RR$ the following identity holds,
\begin{align} \label{ridge.eq.generator.averaging}
\int_{u \in \RR} \, \A\phi(x,u) \, e^{B(x,u)} \, du \;=\; \Gen \phi(x)
\end{align}
where $\Gen$ is the generator of the limiting diffusion \eqref{ridge.eq.limit.diff}, 
$\phi \in \mathcal C$ is a test function in the core of $\Gen$ and
$\A\phi(x,u)$ reads
\begin{align}
\A\phi(x,u) 
= \EE\,\big[\, F'(D B)  \big( A'(x) + \partial_x B(x,u + Z) \big)\,\big] \, \phi'(x)
+ \tfrac{1}{2} \, \EE\,\big[\,F(D B)\,\big] \, \phi^{''}(x)
\end{align}
where $D B = B(x,u+Z)-B(x,u)$ and $Z \sim\Normal(0,1)$.
The proof is a routine calculation that is based on the symmetry of the Gaussian distribution and the fact that the accept-reject function  $F$ verifies the reversibility condition \eqref{ridge.eq.reversibility}.
More specifically, the derivative of equation  \eqref{ridge.eq.reversibility} also reads 
\begin{align} \label{ridge.eq.reversibility.derivative}
F(r) = F'(r) + e^r F'(-r)\ .
\end{align}
This identity also holds for the standard MH function 
$F_{\text{MH}}(r) = \min(1,e^r)$ but has to be interpreted in the sense of distributions. 
In the scalar case $n_x=1$ with $\ell=1$, the generator 
of \eqref{ridge.eq.limit.diff} reads
$$\Gen \phi(x) = \tfrac12 \, \big( a_0(x) \, A'(x)+ a'_0(x)  \big) \, \phi'(x) 
+ \tfrac{1}{2}\, a_0(x) \, \phi^{''}(x)$$ 
where $a_0(x) \equiv a_0(x,1)$ is the mean acceptance probability
$ a_0(x) = \int_{u \in \RR} \EE\,[\,F(D B)\,] \, e^{B(x,u)} \, du$.
To prove \eqref{ridge.eq.generator.averaging} it suffices to verify that
\begin{align} \label{ridge.eq.averaging.identity}
\EE\,\big[\, F'(D B) \,\partial_x B(x,u + Z)\,\big]=\tfrac12 \, a'_0(x) 
\quad \textrm{and} \quad
\EE\,\big[\,F'(D B)\, \big]= \tfrac12\, a_0(x)\ .
\end{align}
Let us prove that the first identity holds. Assumption \ref{ridge.assump.pi} 
justify the following derivation under the integral sign,
\begin{align*}
\partial_x a_0(x) =\int \EE\,\big[\,F'(D B) \, \big(\partial_x B(x,u+Z) - \partial_x B(x,u) \big)
+ F(D B) \partial_x B(x,u)\,\big] \, e^{B(x,u)} \, du\ .
\end{align*}
Equation \eqref{ridge.eq.reversibility.derivative} shows that $F(D B) = F'(D B) + F(-D B)e^{D B}$; since $e^{D B} e^{B(x,u)} = e^{B(x,u+Z)}$, we have
\begin{align*}
\partial_x a_0(x) 
&=\int \EE\,\big[\, F'(D B) \partial_x B(x,u+Z) e^{B(x,u)} \,\big] \, du
+\int \EE\,\big[\, F'(-D B) \partial_x B(x,u) e^{B(x,u+Z)} \,\big] \, du\ .
\end{align*}
The symmetry of the Gaussian distribution $Z \sim \Normal(0,1)$ 
then shows that 
\begin{align}
\int \EE\,\big[\,F'(D B) \partial_x B(x,u+Z) e^{B(x,u)}\,\big] \, du
= \int \EE\,\big[\,F'(-D B) \partial_x B(x,u) e^{B(x,u+Z)} \big]\ .
\end{align}
This concludes the proof 
of the first identity of \eqref{ridge.eq.averaging.identity}. The proof of the second 
identity is similar and even simpler, and thus omitted.

\bibliographystyle{plain}
\bibliography{references}     

\end{document}